\pdfoutput=1
\documentclass{CSML}

\def\dOi{12(1:7)2016}
\lmcsheading%
{\dOi}
{1--32}
{}
{}
{Sep.~10, 2015}
{Mar.~31, 2016}
{}

\ACMCCS{[{\bf Theory of computation}]: Formal languages and automata theory}

\usepackage{amsmath,amssymb}
\usepackage{microtype}
\usepackage{tikz}
\usepackage{wrapfig}

\usepackage{hyperref}

\usetikzlibrary{automata,positioning,arrows}
\tikzset{automaton/.style={node distance=2cm,on grid}}
\tikzset{every state/.style={minimum size=0pt,inner sep=1pt}}
\tikzset{transition/.style={->,>=stealth',shorten >=1pt}}
\tikzset{every initial by arrow/.style={transition}}
\tikzset{initial text={}}
\tikzset{
    >=stealth',
    punkt/.style={
           rectangle,
           rounded corners,
           draw=black, thick,
           text width=6.5em,
           minimum height=2em,
           text centered},
    pil/.style={
           ->,
           shorten <=2pt,
           shorten >=2pt,},
    qwe/.style={
           -,
           shorten <=2pt,
           shorten >=2pt,}
}
\usepackage{thmtools}   % especially nice for restating theorems
\declaretheorem[within=section]{theorem}
\declaretheorem[sibling=theorem]{corollary}
\declaretheorem[sibling=theorem]{lemma}
\declaretheorem[sibling=theorem]{proposition}
\declaretheorem[sibling=theorem,style=definition]{definition}
\declaretheorem[sibling=theorem,style=definition]{example}

\newcommand{\cutout}[1]{}
\newcommand{\ifdraft}[2]{#2}  % #1 for draft; #2 for nondraft

\newcommand{\noterg}[2]{\textcolor{darkgray}{[\textcolor{red}{#1}: #2]}}

\newcommand{\rg}[1]{\noterg{RG}{#1}}
\newcommand{\ntside}[1]{\marginpar{\parbox{12mm}{\color{blue}\raggedright\tiny #1}}}

\newcommand\card[1]{|#1|}
\newcommand\defeq{\triangleq}

\newcommand\too[1][]{\overset{#1}{\mathrel{{\rlap{$\rightarrow$}\,\mathord{\rightarrow}}}}}
\newcommand\boldemph[1]{\emph{\textbf{#1}}}

\renewcommand\AA{\mathcal{A}}

\newcommand\FF{\mathcal{F}}
\newcommand\PP{\mathcal{P}}
\newcommand\PPfn{\mathcal{P}_{\!\mathsf{fn}}}
\newcommand\PPnz{\mathcal{P}_{\!\not=\emptyset}} % TODO: use \empti ?

\newcommand\LL{\mathcal{L}}

\newcommand\N{\mathbb{N}}
\newcommand\Z{\mathbb{Z}}
\newcommand\lab{\mathsf{Lab}}

\newcommand\trdelta[2][]{\overset{#2}\longrightarrow}
\newcommand\xtrdelta[2][]{\xrightarrow{#2}}

\newcommand\trdeltaa[2][]{\overset{#2}{\mathrel{{\rlap{$\longrightarrow$}\,\mathord{\longrightarrow}}}}}
\newcommand\trans[1]{\overset{#1}\longrightarrow}

\newcommand\transs[1]{\overset{#1}{\mathrel{{\rlap{$\longrightarrow$}\,\mathord{\longrightarrow}}}}}
\newcommand\xrightt[1]{\mathrel{\xrightarrow{#1}\!\!\!\!\!\to}}
\newcommand\range[1]{\mathsf{img}(#1)}

\newcommand\id{\mathsf{id}}

\newcommand\move[2]{#1\mathbin{\mathsf{in}}#2}

\newcommand\names{\mathcal{N}}

\newcommand\his{\mathsf{Asn}}

\newcommand\ta{,}

\newcommand\empti{\varnothing}
\newcommand\at{@}
\newcommand\fix[2]{#1\,\mathsf{fix}\,#2}

\title[History-Register Automata]{History-Register Automata}

\author[R.~Grigore]{Radu Grigore\rsuper a}
\address{{\lsuper a}University of Kent}
\email{radugrigore@gmail.com}

\author[N.~Tzevelekos]{Nikos Tzevelekos\rsuper b}
\address{{\lsuper b}Queen Mary University of London}
\email{nikos.tzevelekos@qmul.ac.uk}

\begin{document}

\begin{abstract}
Programs with dynamic allocation are able to create and use an unbounded number of fresh resources, such as references, objects, files, etc.
We propose History-Register Automata (HRA), a new automata-theoretic formalism for modelling such programs.
HRAs extend the expressiveness of previous approaches and bring us to the limits of decidability for reachability checks.
The distinctive feature of our machines is their use of unbounded memory sets (histories) where input symbols can be selectively stored and compared with symbols to follow. In addition, stored symbols can be consumed or deleted by reset.
We show that the combination of consumption and reset capabilities renders the automata powerful enough to imitate counter machines,
and yields closure under all regular operations apart from complementation.
We moreover examine weaker notions of HRAs which strike different balances between expressiveness and effectiveness.
\end{abstract}

\maketitle % lmcs warns if this comes before abstract

\section{Introduction}\label{sec:intro}
%input{introduction.tex}
%
Program analysis faces substantial challenges due to its aim to devise finitary methods and machines which are required to operate on potentially infinite program computations.
A specific such challenge stems from dynamic generative behaviours such as, for example, object or thread creation in Java, or reference creation in ML\null. A program engaging in such behaviours is expected to generate a possibly unbounded amount of distinct resources, each of which is assigned a unique identifier, a \emph{name}. Hence, any machine designed for analysing such programs is expected to operate on an infinite alphabet of names.
The latter need has brought about the introduction of automata over infinite alphabets in program analysis, starting from prototypical machines for mobile calculi~\cite{HDA} and variable programs~\cite{RA1}, and recently developing towards automata for verification tasks such as equivalence checks of ML programs~\cite{ML1,ML2}, context-bounded analysis of concurrent programs~\cite{Java1,Java2} and runtime program monitoring~\cite{TOPL}.

The literature on automata over infinite alphabets is rich in formalisms each based on a different approach for tackling the infiniteness of the alphabet in a finitary manner (see~e.g.~\cite{Segoufin_overview} for an overview). A particularly intuitive such model is that of \emph{Register Automata (RA)}~\cite{RA1,RA2}, which are machines built around the concept of an ordinary finite-state automaton attached with a fixed finite amount of registers. The automaton can store in its registers names coming from the input, and make control decisions by comparing new input names with those already stored. Thus, by talking about addresses of its memory registers rather than actual names, a so finitely-described automaton can tackle the infinite alphabet of names.
Driven by program analysis considerations, register automata have been recently extended with the feature of name-freshness recognition~\cite{FRA}, that is, the capability of the automaton to accept specific inputs just if they are \emph{fresh}\,---\,they have not appeared before during computation.
Those automata, called \emph{Fresh-Register Automata (FRA)}, can account for languages like the following,
\[
\LL_0 =\{a_1\ldots a_n\in\names^*\ |\ \forall i\not=j.\ a_i\not=a_j\}
\]
which captures the output of a fresh-name generator ($\names$ is a countably infinite set of names). 
FRAs are expressive enough to model, for example, finitary fragments of languages like the $\pi$-calculus~\cite{FRA} or ML~\cite{ML1}. 

The freshness oracle of FRAs administers the automata with perhaps too restricted an access to the full history of the computation: it allows them to detect name freshness, but not non-freshness.
Consider, for instance, the following simple language,
\begin{align*}
\LL'=\{w\in(\{O,P\}\times\names)^*\ |\ &\text{each letter of $w$ appears exactly once in it}\\
	&\land \text{ each $(O,a)$ in $w$ is preceded by some $(P,a)$}\, \}
\end{align*}
where the alphabet is made of pairs containing an element from the set $\{O,P\}$ and a name ($O$ and $P$ can be seen as different processes or agents exchanging names). The language $\LL'$
represents a paradigmatic scenario of a name generator $P$ coupled with a name consumer $O$: each consumed name must have been created first, and no name can be consumed twice. It can capture e.g.~the interaction of a process which creates new files with one that opens them, where no file can be opened twice.
\cutout{; or, see~\cite{Dan}, that of a program which can pass control to its environment through continuation names, where the environment must return one of the continuations already created but not returned yet (so as e.g.~to avoid replay attacks).}
The inability of FRAs to detect non-freshness, as well as the fact that names in their history cannot be removed from it, do not allow them to express $\LL'$. More generally, the notion of \emph{re-usage} or \emph{consumption} of names is beyond the reach of those machines.
Another limitation of FRAs is the failure of closure under concatenation, interleaving and Kleene star.

%
%In particular, they fail to express create-before-use-name behaviours, as present 
%e.g. in reference allocate/write or file create/open scenarios.
%
%3. Another limitation of FRAs is that their history is never emptied. Thus, they 
%cannot capture re-usage of names (e.g. reference de-allocation or file deletion).
%
%The above shortcomings are addressed by HRAs.
%

\begin{figure}[t]
\begin{center}
\parbox{.28\linewidth}{
\begin{tikzpicture}[automaton]
\node[state,initial,above,accepting] (q0) {$q_0$};
\node[state,right=of q0] (P) {$P$};
\node[state,below=of q0] (O) {$O$};
\path[transition]
  (q0) edge[pil,above] node{$P$} (P)
  (q0) edge[pil,left] node{$O$} (O)
  (P) edge[pil,bend left,below] node{$\empti\ta 1$} (q0)
  (O) edge[pil,bend right,right] node{$1\ta 2$} (q0);
\end{tikzpicture}}
\parbox{.7\linewidth}{%
The automaton starts at state $q_0$ with empty history and non-deterministically makes a transition to state $P$ or $Q$, accepting the respective symbol. From state $P$, it accepts any input name $a$ which does not appear in any of its histories (this is what $\empti$ stands for), puts it in history number 1, and moves back to $q_0$. From state $O$, it accepts any input name $a$ which appears in history number 1, puts it in history number 2, and moves back to $q_0$.}%\vspace{-2.5mm}
\end{center}
\caption{History-register automaton accepting $\LL'$.}%\vspace{-2mm}
\label{fig:HRA}
\end{figure}

Aiming at providing a stronger theoretical tool for analysing computation with names, in this work we further capitalise on the use of histories by effectively upgrading them to the status of registers. That is, in addition to registers, we equip our automata with a fixed number of unbounded sets of names (\emph{histories}) where input names can be stored and compared with names to follow. As histories are internally unordered, the kind of name comparison we allow for is name belonging (\emph{does the input name belong to the $i$-th history?}). 
Moreover, names can be selected and removed from histories, and individual histories can be emptied/reset. 
We call the resulting machines \emph{History-Register Automata (HRA)}. 
For example, $\LL'$ is accepted by the HRA with 2 histories depicted in 
\autoref{fig:HRA}, where by convention we model pairs of symbols by sequences of two symbols.\footnote{Although, technically speaking, the machines we define below do not handle constants (as e.g.~$O,P$), the latter are encoded as names appearing in initial registers, in standard fashion.} 

The strengthening of the role of histories substantially increases the expressive power of our machines.
%
%Different input names may be stored in distinct histories and checked for different properties; or, individual names can be removed from histories, thus allowing us to express \emph{consumption} of resources.
More specifically, we identify three distinctive features of HRAs:
\begin{enumerate}
\item[(1)] the capability to reset histories;
\item[(2)] the use of multiple histories;
\item[(3)] the capability to select and remove individual names from histories.
\end{enumerate}
Each feature allows us to express one of the paradigmatic languages below, none of which are FRA-recognisable. 
%(note that $\LL_3$ is a simplified version of $\LL'$). 
\begin{align*}
\LL_1 &= \{a_0w_1\ldots a_0w_n \in\names^* |\ \forall i.\ w_i\in\names^*\land a_0w_i\in\LL_0\} \;\;\text{for given $a_0$}\\[1mm]
\LL_2 &= \{a_1a_1'\ldots a_na_n'\in\names^* |\ a_1\ldots a_n,a_1'\ldots a_n'\in\LL_0\} \\[1mm]
\LL_3 &= \{a_1\ldots a_na_1'\ldots a_{n'}'\in\names^* |\ a_1\ldots a_n,a_1'\ldots a_{n'}'\in\LL_0\land \forall i.\exists j.\,a_i'=a_j\} 
%\LL_3 &= \{a_1\ldots a_n\in\names^* |\ \text{each $a_i$ appears at most twice}\,\}
\end{align*}
%are used here to demonstrate the limitations of FRAs: although the latter are sufficiently powerful for expressing the semantics of programs with generative effects, they fall short of providing a satisfactorily rich verification toolkit for them. We cannot (a)~Kleene-close FRAs (actually, not even concatenate them), (b)~interleave them nor (c)~use them to express {non-freshness} or {consumption} of names.
%
%Note that the language $\LL_3$ is a simplified version of $\LL'$.
%
Apart from the gains in expressive power, the passage to HRAs yields a more well-rounded automata-theoretic formalism for generative behaviours as these machines enjoy closure under all regular operations apart from complementation.
\cutout{Moreover, the use of several histories 
allows one to express interleavings of languages as different histories can be used for each of them (in the same automaton). Using a similar technique we can also simulate behaviours where an automaton can drop the first input symbol if the input string cannot lead to an accepting state. Such \emph{rollback} behaviours are essential for runtime checks~\cite{TOPL}.}%
On the other hand, the combination of features (1-3) above enables us to use histories as counters and simulate counter machines. 
\cutout{In particular, we are able to reduce coverability for reset Petri nets to language emptiness for HRAs, and vice versa.}%
We therefore obtain non-primitive recursive bounds for checking language emptiness.
Given that language containment and universality are undecidable already for register automata~\cite{RA2}, HRAs are fairly close to the decidability boundary for properties of languages over infinite alphabets.
Nonetheless, starting from HRAs and weakening them in each of the first two factors (1,2) we obtain automata models which are still highly expressive but computationally more tractable.
Overall, the expressiveness hierarchy of the machines we examine is depicted in \autoref{fig:hierarchy} (weakening in (1) and (2) respectively occurs in the second column of the figure).

\begin{figure}[t]
\begin{center}
\begin{tikzpicture}[node distance=8mm, auto,]
 %nodes
 \node[punkt] (HRA) {HRA};
 \node[punkt, left=of HRA] (unHRA) {unary HRA}
    edge[pil] (HRA);  
 \node[punkt, below=of unHRA] (nrHRA) {non-reset HRA}
    edge[pil] (HRA);
 \node[punkt, right=of nrHRA] (CMA) {\em DA\,/\,CMA};    
 \node[punkt, left=of unHRA] (FRA) {\em FRA}
    edge[pil] (unHRA)
    edge[pil] (nrHRA);
 \node[punkt, left=of nrHRA] (RA) {\em RA}
    edge[pil] (FRA);
 \path (nrHRA) edge[pil] (CMA);
\end{tikzpicture}%\vspace{-3mm}
\end{center}
\caption{Expressiveness of history-register automata compared to previous models (in italics). The inclusion $\mathcal{M}\longrightarrow\mathcal{M}'$ means that for each $\AA\in\mathcal{M}$ we can effectively construct an $\AA'\in\mathcal{M}'$ accepting the same language as $\AA$. All inclusions are strict.}%\vspace{1.5mm}
%\hrule\vspace{-4mm}
\label{fig:hierarchy}
\end{figure}

\subsection*{Motivation and related work.}
The motivation for this work stems from semantics and verification. In semantics, the use of names to model resource generation originates in the work of Pitts and Stark on the $\nu$-calculus~\cite{nu} and Stark's PhD~\cite{Stark:PhD}.
Names have subsequently been incorporated in the semantics literature (see e.g.~\cite{nom1,nom2,nom3,nom4}), especially after 
the advent of \emph{Nominal Sets}~\cite{nom}, which provided formal foundations for doing mathematics with names. Moreover,
recent work in game semantics has produced algorithmic representations of game models using extensions of fresh-register automata~\cite{ML1,ML2,IMJ2}, thus achieving automated equivalence checks for fragments of ML and Java. 
In a parallel development, a research stream on automated analysis of dynamic concurrent programs has developed essentially the same formalisms, this time stemming from trace-based operational  techniques~\cite{Java1,Java2}.
This confluence of different methodologies is exciting and encourages the
development of stronger automata for a wider range of verification tasks, and just such an automaton we examine herein.

Although our work is driven by program analysis, the closest existing automata models to ours come from XML database theory and model checking. Research in the latter area has made great strides in the last years on automata over infinite alphabets and related logics (e.g.~see~\cite{Segoufin_overview} for an overview from 2006).
As we show in this paper, history-register automata fit very well inside the big picture of automata over infinite alphabets (cf.~\autoref{fig:hierarchy}) and in fact can be seen as closely related to \emph{Data Automata (DA)}~\cite{DA} or, equivalently, \emph{Class Memory Automata (CMA)}~\cite{CMA}. 
A crucial difference lies in the reset capabilities of our machines, which allow us to express languages like $\LL_1$ that cannot be expressed by DA/CMAs.
On the other hand, the local termination conditions of DA/CMAs allow them to express languages
that HRAs cannot capture.
We find the correspondence between HRAs and DAs particularly pleasing as it relates two seemingly very different kind of machines, with distant operational descriptions and intuitions.
%This fit leaves space for transfer of technologies and, more specifically, of the associated logics of data automata.

A recent strand of research in foundations of atom-based computation~\cite{WNom1,WNom2,WNom3} has examined nominal variants of classical machine models, ranging from finite-state automata to Turing machines.
% Based on the interpretations carried out in this paper, 
%we conjecture that
%HRAs can be seen as nominal versions of vector addition systems (with states and possibly some tailor-made restrictions).
Finally, since the publication of the conference version of this paper~\cite{CONF}, there has been work in \emph{nested DA/CMAs}~\cite{NDA,NCMA}, which can be seen as extensions of non-reset HRAs whereby the histories satisfy some nesting relations. The latter are a clean extension of our machines, leading to higher reachability complexities.
%\ntside{make this more precise}

This article is the journal version of~\cite{CONF},
  with strengthened results and with full proofs.
\autoref{sec:registers} is new:
  it collects all results that show how registers can be simulated by other means.
Many upper bounds are tighter:
  Propositions
  \ref{prop:regs_his},
  \ref{prop:nrHRA-regs},
  \ref{prop:emptiness-ub},
  \ref{prop:nrHRA-ub},
  \ref{prop:unary-empty-ub}.
Some results are new:
  Propositions
  \ref{prop:empty-init-histories},
  \ref{prop:remove-regs-withstates},
  \ref{prop:emptiness-general},
  \ref{prop:unary-empty}.
\autoref{ex:HRA-regs-yes} is new.
Most proofs have been revised.
\autoref{sec:summary} is new:
  it collects in one place the main properties of HRAs.
Also, we relate our work to what has been done after the conference version was published.
%\rg{Feedback from Hongseok:
%`Reviewers appreciate if the differences from the conference version are described in detail.'}

\subsection*{Overview}
In \autoref{sec:defns} we introduce HRAs and their basic properties.
In \autoref{sec:closure} we examine regular closure properties of HRAs.
In \autoref{sec:registers} we explain how registers can be simulated by other means,
  such as histories.
In \autoref{sec:empty} we prove that emptiness is $\textsc{Ackermann}$-complete.
In \autoref{sec:weak} we introduce weaker models, and study their properties.
In \autoref{sec:summary} we summarize the main properties of HRAs.
In \autoref{sec:connect} we connect HRAs to existing automata formalisms.
We conclude by discussing future directions which emanate from this work.

%%% Local Variables: 
%%% mode: latex
%%% TeX-master: "jour"
%%% End: 

\section{Definitions and first properties}\label{sec:defns}

We start by fixing some notation. Let $\names$ be a countably infinite alphabet of \emph{names}, over which we range by $a$,~$b$, $c$, etc. For any pair of natural numbers $i\leq j$, we write $[i,j]$ for the set $\{i,i{+}1,\ldots,j\}$, and for each $i$ we let $[i]$ be the set $\{1,\ldots,i\}$.
For any set $S$,
  we write $\card{S}$ for the cardinality of $S$,
  we write $\PP(S)$ for the powerset of $S$,
  we write $\PPfn(S)$ for the set of finite subsets of $S$, and
  we write $\PPnz(S)$ for the set of nonempty subsets of $S$.
We write $\id:S\to S$ for the identity function on $S$, and $\range{f}$ for the image of $f:S\to T$. 

We define automata which are equipped with a fixed number of \boldemph{registers} and \boldemph{histories} where they can store names.
Each register is a memory cell where one name can be stored at a time; each history can hold an unbounded set of names.
We use the term \boldemph{place} to refer to both histories and registers.
\cutout{\footnote{We will see below that, in fact, registers are subsumed by histories in the case of machines with resets, but this does not hold in general.}}
Transitions are of two kinds: name-accepting transitions and reset transitions. Those of the former kind have labels of the form $(X,X')$, for sets of places $X$ and $X'$; and those of the latter carry labels with single sets of places $X$. A transition labelled $(X\ta X')$ means:
\begin{itemize}
\item accept name $a$ if it is contained precisely in places $X$, and
\item update places in $X$ and $X'$ so that $a$ be contained precisely in places $X'$ after the transition (without touching other names).
\end{itemize}
By $a$ being contained precisely in places $X$ we mean that it appears in every place in $X$, and in no other place.
In particular, the label $(\empti,X')$ signifies accepting a fresh name (one which does not appear in any place) and inserting it in places $X'$.
On the other hand,
a transition labelled by $X$ resets all the places in $X$;
  that is, it updates each of them to the empty set (registers are modelled as sets with at most one element).
Reset transitions do not accept names; they are $\epsilon$-transitions from the outside.
Note that the label $(X\ta\empti)$ has different semantics from the label $X$:
the former stipulates that a name appearing precisely in $X$ be accepted and then removed from $X$; whereas the latter clears all the contents of places in $X$, without accepting anything.

\subsection{Definitions}

Formally, {let us fix positive integers $m$ and $n$ which will stand for the default number of histories and registers respectively in the machines we define below.} The set $\his$ of \boldemph{assignments} and the set $\lab$ of \boldemph{labels} are:
\begin{align*} 
\his &= \{\,H:[m+n]\to\PPfn(\names) \mid \forall i>m.\,|H(i)|\leq 1\,\}\\
\lab &= \PP([m+n])^2 \cup \PP([m+n])
\end{align*}
For example, $\{(i,\empti)\ |\ i\in[m{+}n]\}$ is the empty assignment.\footnote{We represent functions as sets of pairs.}
We range over elements of $\his$ by $H$ and variants, and over elements of $\lab$ by {$\ell$} and variants.

%For any assignment $H$ and any $a\in\names$, $S\subseteq\names$ and $X\subseteq[m{+}n]$:
% RG: i replaced the above mainly to avoid two formulas separated only by punctuation
Let $H \in \his$ be an assignment,
  let $a \in \names$ be a name,
  let $S \subseteq \names$ be a set of names, and
  let $X \subseteq [m+n]$ be a set of places.
We introduce the following notation:
\begin{itemize}
\item We set $H\at X$ to be the set of names which \emph{appear precisely} in places $X$ in $H$;
  that is, $H\at X=\bigcap_{i\in X}\!H(i)\setminus \bigcup_{i\notin X}\!H(i)$.
In particular, $H\at\,\empti=\names\setminus \bigcup_{i}\!H(i)$ is the set of names which do not appear in $H$.
\item $H[X\mapsto S]$ is the update $H'$ of $H$ so that all places in $X$ are mapped to $S$;
  that is, \penalty-100 $H'\!=\!\{(i,H(i))\ |\ i\not\in X\}\cup\{(i,S)\ |\ i\in X\}$.
  E.g., $H[X\mapsto\empti]$ resets all places in~$X$.
\item $H[\move{a}{X}]$ is the update of $H$ which removes name $a$ from all places and inserts it back in~$X$;
   that is, $H[\move{a}{X}]$ is the assignment: %for all $i$:
%\[
%H[\move{a}{X}](i)=
%\begin{cases}
%H(i)\setminus\{a\} & i\notin X\\
%H(i)\cup\{a\} & i\in X\cap[m]\\
%\{a\} & i\in X\setminus[m]
%\end{cases}
%\]
%\vspace{-1mm}
\[
\{ (i,H(i)\cup\{a\})\ |\ i\in X\cap[m]\}
\ \cup \
\{ (i,\{a\})\ |\ i\in X\setminus[m]\}
\ \cup \
\{ (i,H(i)\setminus\{a\})\ |\ i\notin X\}
\]
\end{itemize}
Note above that operation $H[\move{a}{X}]$ acts differently in the case of histories ($i\leq m$) and registers ($i>m$) in $X$: in the former case, the name $a$ is added to the history $H(i)$, while in the latter the register $H(i)$ is set to $\{a\}$ and its previous content is cleared.

We can now define our automata.
%\dd{I don't see how the automata definition depends from $m$. Maybe we should make the connection explicit, for example by parameterising
%the $\his$ into $\his [m]$ so in the definition below we would have $H_0 \in \his [m]$? }
%\nt{I added a sentence above. We could use indexing by $m$ but then we'd need to be consistent and add it also to $\lab$ (possibly elsewhere) and clutter notation.}

\begin{definition}
A \boldemph{history-register automaton (HRA)} of type $(m,n)$ is a tuple $\AA=\langle Q,q_0,H_0,\delta,F\rangle$ where:
\begin{itemize}
\item $Q$ is a finite set of states, $q_0$ is the initial state, $F\subseteq Q$ are the final ones,
\item $H_0\in\his$ \ is the initial assignment, and
\item $\delta\subseteq Q\times\lab\times Q$ \ is the transition relation.
\end{itemize}
For brevity, we shall call $\AA$ an \emph{$(m,n)$-HRA}.
\end{definition}%\vspace{-1mm}

We write transitions in the forms $q\trans{X,X'}q'$ and $q\trdelta{X}q'$, for each kind of transition label.
\cutout{We can distinguish a few types of transitions of the first type, according to their label~$X\ta X'$.
If $X=\empti$ we say the transition is \emph{fresh};
if $X'=\empti$ we say it is a \emph{deletion};
If $X=X'$ we say it is a \emph{test};
and if $X\subseteq X'$ we say it is \emph{ascending}.}%
In diagrams,
we may unify different transitions with common source and target, for example $q\trdelta{X,X'}q'$ and $q\trdelta{Y,Y'}q'$ may be written \hbox{$q\xtrdelta{X,X'\,/\,Y,Y'}q'$;} moreover, we shall lighten notation and write $i$ for the singleton $\{i\}$, and $ij$ for $\{i,j\}$.

We already gave an overview of the semantics of HRAs.
This is formally defined by means of configurations representing the current computation state of the automaton.
A \boldemph{configuration} of $\AA$ is a pair $(q,H)\in\hat{Q}$, where:%\vspace{-1mm}
\[ \hat{Q}=Q\times\his
%\vspace{-1mm}
\]
From the transition relation~$\delta$ we obtain the configuration graph of $\AA$ as follows.
%Below we denote the empty sequence by~$\epsilon$.
%
\begin{definition}\label{def:confgraph}
Let $\AA$ be an $(m,n)$-HRA as above.
Its \boldemph{configuration graph} \hbox{$(\hat{Q},\trans{})$,} where \  ${\trans{}}\,\subseteq\, \hat Q\times\bigl(\names\cup\{\epsilon\}\bigr)\times\hat Q$, is constructed by setting $(q,H)\trans{x}(q',H')$ if and only if one of the following conditions is satisfied.
\begin{itemize}
\item $x=a\in\names$ and there is $q\trdelta{X\ta X'}q'\;\in\;\delta$ such that $a\in H\at X$ and $H'=H[\move{a}{X'}]$.
\item $x=\epsilon$ and there is $q\trdelta{X}q'\;\in\;\delta$ such that $H'=H[X\mapsto\empti]$.
\end{itemize}
The language accepted by $\AA$ is 
\[
\LL(\AA) = \{\,w\in\names^*\mid
  \text{$(q_0,H_0)\transs{w}(q,H)$ and $q\in F$}\,\}
\]
where \ $\transs{}$ \ is the reflexive transitive closure of \ $\trans{}$ \ (i.e.\ $\hat q\xrightt{x_1\ldots x_k}\hat q'$ if $\hat
q\trans{x_1}\!\cdots\!\trans{x_k}\hat q'$).
\end{definition}

Note that we use $\epsilon$ both for the empty sequence and the empty transition so, in particular, when writing sequences of the form $x_1\ldots x_k$ we may implicitly consume $\epsilon$'s.
It is worth noting here that our 
formulation follows \emph{M-automata}~\cite{RA1} in that multiple places can be mentioned at each transition.

\cutout{
\rg{TODO: make this fit in text}
 Another design choice regards the use of sets of places in transitions instead e.g.~of single places. Although the latter description would lead to an equivalent and probably conciser formalism, it would be inconvenient for combining HRAs e.g.~in order to produce the intersection of their accepted languages. In fact, our formulation follows \emph{M-automata}~\cite{RA1}, an equivalent presentation of RAs susceptible to closure constructions.
}

%Note also that reflexivity of $\transs{}$ implies $\hat{q}\transs{\epsilon}\hat{q}$.
%by writing, for example, $\epsilon$ for $\epsilon\ldots\epsilon$, or $a$ for $a\epsilon\ldots\epsilon$, etc.

\begin{example}\label{ex:HRA}
The language $\LL_1$ of the Introduction is recognised by the following $(1,1)$-HRA (leftmost below), with initial assignment $\{(1,\empti),(2,a_0)\}$. The automaton starts by accepting $a_0$, leaving it in register $2$, and moving to state $q_1$. There, it loops accepting fresh names (appearing in no place) which it stores in history $1$. From $q_1$ it goes back to $q_0$ by resetting its history.
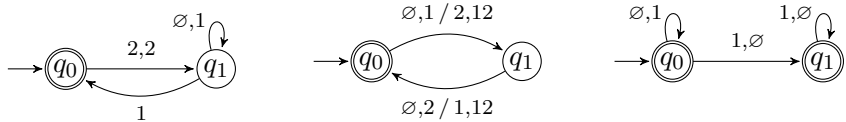
\begin{center}%\vspace{-2.5mm}
\begin{tikzpicture}[automaton]
\node[state,initial,accepting] (q0) {$q_0$};
\node[state] (q1) [right=of q0] {$q_1$};
\path[transition]
  (q0) edge node[above]{$\scriptstyle 2,2$} (q1)
  (q1) edge[bend left] node[below]{$\scriptstyle 1$} (q0)
  (q1) edge[loop above] node[left]{$\scriptstyle \empti,1$} ();
\end{tikzpicture}\qquad
\begin{tikzpicture}[automaton]
\node[state,initial,accepting] (q0) {$q_0$};
\node[state] (q1) [right=of q0] {$q_1$};
\path[transition]
  (q0) edge[bend left] node[above]{$\scriptstyle\empti,1\,/\,2,12$} (q1)
  (q1) edge[bend left] node[below]{$\scriptstyle\empti,2\,/\,1,12$} (q0);
\node[state,initial,accepting] (q0a) [right=of q1] {$q_0$};
\node[state,accepting] (q1a) [right=of q0a] {$q_1$};
\path[transition]
  (q0a) edge node[above]{$\scriptstyle1\ta\empti$} (q1a)
       edge[loop above] node[left]{$\scriptstyle\empti\ta 1$} ()
  (q1a) edge[loop above] node[left]{$\scriptstyle1\ta\empti$} ();
\end{tikzpicture}
%\vspace{-3.5mm}
\end{center}
We can also see that the other two HRAs, of type $(2,0)$ and $(1,0)$, accept the languages $\LL_2$ and $\LL_3$ respectively. 
Both automata start with empty assignments.

Finally, the automaton we drew in \autoref{fig:HRA} is, in fact, a (2,2)-HRA where its two registers initially contain the names $O$ and $P$ respectively. The transition label $O$ corresponds to $(3,3)$, and $P$ to~$(4,4)$. 
\end{example}

As mentioned in the introductory section, HRAs build upon \emph{(Fresh) Register Automata}~\cite{RA1,RA2,FRA}. The latter can be defined within the HRA framework as follows.\footnote{The definitions given in~\cite{RA1,RA2,FRA} are slightly different but  can routinely be shown equivalent.}

\begin{definition}
A \emph{Register Automaton (RA)} of $n$ registers is a $(0,n)$-HRA with no reset transitions.
A \emph{Fresh-Register Automaton (FRA)} of $n$ registers is a $(1,n)$-HRA $\AA=\langle Q,q_0,H_0,\delta,F\rangle$ such that $H_0(1)=\bigcup_{i}H_0(i)$ and:
\begin{itemize}
\item for all $(q,\ell,q')\in\delta$, there are $X,X'$ such that $\ell=(X,X')$ and $1\in X'$;
\item for all $(q,\{1\},X',q')\in\delta$, there is also $(q,\empti,X',q')\in\delta$.
\end{itemize}
\end{definition}

Thus, in an FRA all the initial names must appear in its history, and the same holds for all the names the automaton accepts during computation ($1\in X'$). As, in addition, no reset transitions are allowed, the history effectively contains all names of a run. On the other hand, the automaton cannot recognise \emph{non-freshness}: if a name appearing only in the history is to be accepted at any point then a totally fresh name can be also be accepted in the same way. Now, from~\cite{FRA} we have the following. 

\begin{lemma}\label{lem:FRA}
The languages $\LL_1,\LL_2$ and $\LL_3$ are not FRA-recognisable.
\end{lemma}

\begin{proof}
$\LL_1$ was explicitly examined in~\cite{FRA}. For $\LL_2$ and $\LL_3$ we use a similar argument as the one for showing that $\LL_0*\LL_0$ is not FRA-recognisable~\cite{FRA} .
\end{proof}

\subsection{Bisimulation}
Bisimulation equivalence, also called \emph{bisimilarity}, is a useful tool for relating automata, even from different paradigms. It implies language equivalence and is generally easier to reason about than the latter. We will be using it avidly in the sequel.

\begin{definition}\label{d:bisim}
Let $\AA_i=\langle Q_i,q_{0i},H_{0i},\delta_i,F_i\rangle$ be $(m,n)$-HRAs, for $i=1,2$.
A relation $R\subseteq\hat Q_1\times\hat Q_2$ is called a \emph{simulation} on $\AA_1$ and $\AA_2$ if, for all $(\hat{q}_1,\hat{q}_2)\in R$,%\vspace{1mm}
\begin{itemize}
\item if $\hat q_1\trdeltaa[\delta_2]{\epsilon}\hat q_1'$ and $\pi_1(\hat q_1')\in F_1$ then $\hat q_2\trdeltaa[\delta_2]{\epsilon}\hat q_2'$ for some $\pi_1(\hat q_2')\in F_2$, where $\pi_1$ is the first projection function;
\item if $\hat q_1\trdeltaa{\epsilon}\!\cdot\!\trdelta[\delta_1]{a}\hat q_1'$ then $\hat q_2\trdeltaa{\epsilon}\!\cdot\!\trdelta[\delta_2]{a}\hat q_2'$ for some $(\hat q_1',\hat q_2')\in R$.%\vspace{1mm}
\end{itemize}
$R$ is called a \boldemph{bisimulation} if both $R$ and $R^{-1}$ are simulations.
We say that $\AA_1$~and~$\AA_2$ are \boldemph{bisimilar}, written $\AA_1\sim\AA_2$, if $((q_{01},H_{01}),(q_{02},H_{02}))\in R$ for some bisimulation $R$.
\end{definition}

%The following is a standard result.

\begin{lemma}\label{l:bisim}
If $\AA_1\sim\AA_2$ then $\LL(\AA_1)=\LL(\AA_2)$.
\end{lemma}

\subsection{Determinism}
We close our presentation here by describing the deterministic class of HRAs.
We defined HRAs in such a way that, at any given configuration $(q,H)$ and for any input symbol $a$, there is at most one set of places $X$ that can match $a$, i.e.~such that $a\in H\at X$. As a result,
the notion of determinism in HRAs can be ensured by purely syntactic means.
Below we write $q\trdeltaa{X}q'\in\delta£$ if there is a sequence of transitions $q\trdelta{X_1}\cdots\trdelta{X_n}q'$ in $\delta$ such that $X=\bigcup_{i=1}^{n}X_i$.
In particular, $q\trdeltaa{\empti}q\in\delta$.

\begin{definition}
We say that an HRA $\AA$ is \boldemph{deterministic} when,
  for any reachable configuration $\hat{q}$ and any name $a$,
    if $\hat{q}\transs{\epsilon}\!\cdot\!\trans{a}\hat{q}_1$ and $\hat{q}\transs{\epsilon}\!\cdot\!\trans{a}\hat{q}_2$ then $\hat{q}_1=\hat{q}_2$.
We say that an HRA $\AA$ is \boldemph{strongly deterministic} when,
  for any state $q$ and any sets $X$, $X_1$, $X_2$, $Y_1$, $Y_2$,
  if $q\trdeltaa{Y_1}\!\cdot\!\xtrdelta{X\setminus Y_1,X_1}q_1\in\delta$ and $q\trdeltaa{Y_2}\!\cdot\!\xtrdelta{X\setminus Y_2,X_2}q_2\in\delta$ then $q_1=q_2$, $Y_1=Y_2$ and $X_1=X_2$.
\end{definition}

Even if $\AA$~is deterministic,
  it is still possible to have multiple paths in the configuration graph
  that are labeled by the same word.
However, such paths may only differ in their $\epsilon$-transitions.
In the definition of `strongly deterministic',
  the set $X$ is guessing where the name~$a$ occurs.

\begin{lemma}
If $\AA$ is strongly deterministic then it is deterministic.
\end{lemma}

%%% Local Variables: 
%%% mode: latex
%%% TeX-master: "jour"
%%% End: 
% vim:spell:spelllang=en_gb:

\section{Closure properties}\label{sec:closure}

History-register automata enjoy good closure properties with respect to regular language operations. In particular, they are closed under union, intersection, concatenation and Kleene star, but not closed under complementation. 

In fact, the design of HRAs is such that the automata for union and intersection come almost for free through a straightforward product construction which is essentially an ordinary product for finite-state automata, modulo reindexing of places to account for duplicate labels (cf.~\cite{RA1}).
\cutout{We do not show these constructions formally (they are very similar to those in~\cite{RA1}), but we demonstrate the second one through an example.

\begin{example}
Break $\LL_2$ into two languages.
\end{example}}%
The constructions for Kleene star and concatenation are slightly more involved as there is need for passing through intermediate automata which do not touch their initial names. 

We shall need the following technical gadget. Given an $(m,n)$-HRA $\AA$ and a sequence $w$ of $k$ distinct names, we construct a bisimilar $(m,n{+}k)$-HRA, denoted $\fix{\AA}{w}$, in which the names of $w$ appear exclusively in the additional $k$ registers, which, moreover, remain unchanged during computation. The construction will allow us, for instance, to create feedback loops in automata ensuring that after each feedback transition the same initial configuration occurs.
%\rg{TODO: say here that we can use also histories}

\begin{lemma}\label{lem:forclos}
Let $\AA$ be an $(m,n)$-HRA with initial assignment $H_0$ and $w=a_1\ldots a_k$ a sequence of distinct names. We can effectively construct an $(m,n{+}k)$-HRA $\fix{\AA}{w}$ with initial assignment $H_0'$ such that $\fix{\AA}{w}\sim \AA$ and:
\begin{itemize}
\item $H_0'(m{+}n{+}i)=a_i$ for all $i\in[k]$, and $H_0'(i)=H_0(i)\setminus\{a_1,\ldots,a_k\}$ for all $i\in[m{+}n]$;
\item for all reachable configurations $(q,H)$ of $\fix{\AA}{w}$ and all $i>m{+}n$, $H(i)=H_0'(i)$.
\end{itemize}
\end{lemma}

\begin{proof}
We construct $\fix{\AA}{w}=\langle Q',q_0',H_0',\delta',F'\rangle$ as follows. First, we insert/move all names of $w$ to the new registers (places $[m{+}n{+}1,m{+}n{+}k]$), i.e.~we set $H_0'(i)=H_0(i)\setminus\{a_1\ldots a_k\}$ for all $i\in[m{+}n]$, and  $H_0'(m{+}n{+}i)=\{a_i\}$ for each $i\in[k]$. 
The role of the new registers is to constantly store the names in $w$ and act on the behalf of other places when the latter intend to use  those names: during computation, whenever an $a_i$ is captured by a transition of the initial automaton $\AA$, in $\fix{\AA}{w}$ it will be instead simulated by a transition involving the new registers. In order for the simulation to be accurate, we shall inject inside states information specifying the \emph{intended} location of the $a_i$s in the places of $\AA$.
Thus, the states of the new automaton are pairs $(q,f)$, where $q\in Q$ and $f$ is a function recording, for each of the new registers, where would the name of the register appear in the original automaton $\AA$. That is,
\[
Q'=Q\times\{f:[k]\to\PP([m{+}n])\ |\ \forall j\not=j'\!.\, f(j)\cap f(j')\subseteq[m]\}
\] 
while $q_0'=(q_0,\{(i,\{j\ |\ a_i\in H_0(j)\})\ |\ i\in[k]\})$
and $F'=\{(q,f)\in Q'\ |\ q\in F\}$. Finally, $\delta'$ operates just like $\delta$ albeit taking into account the $f$'s of states to figure out the intended positions of the $a_i$s and, at the same time, update the $f$'s after each transition. We therefore include in $\delta'$ precisely the following transitions. Below we write $i^\circ$ for $m{+}n{+}i$.
For each $(q,f)\in Q'$ and $q\trdelta{X,X'}q'\in\delta$,%
\begin{itemize}
\item add a transition $(q,f)\trdelta{X,X'}(q',f)$;
\item if $f(i)=X$ for some $i$ then add $(q,f)\xtrdelta{\{i^\circ\},\{i^\circ\}}(q',f')$ where $f'=f[i\mapsto X']$;
\end{itemize}
Moreover, for each $q\trdelta{X}q'\in\delta$ include $(q,f)\trdelta{X}(q',f')$ where $f'=\{(j,f(j)\setminus X)\ |\ j\in[k]\}$.
\\
Following the above line of reasoning, we can show that  the relation
\[
\{ ((q,H),(q,f,H'))\ |\ \forall i\!\in\![m{+}n].\,H(i)=H'(i)\cup\{a_j\,|\,i\in f(j)\}\}
\]
with $(q,H),(q,f,H')$ reachable configurations, is a bisimulation.
\end{proof}

%We now show that languages recognised by HRAs are closed under union, intersection, concatenation and Kleene star.

We write $\LL\circ\LL'$ for concatenation of languages,
  and $\LL^*$ for Kleene closure of a language.
We use the same definitions as the standard ones for languages over finite alphabets:
  $\LL \circ \LL'$ is $\{\,ww'\mid w\in\LL \,\land\, w'\in \LL'\,\}$,
  and $\LL^*$ is the least fixed-point of the
    equations $\epsilon\in\LL^*$ and~$\LL^*\circ\LL \subseteq \LL^*$,
    where $\epsilon$~is the empty word.

\begin{proposition}\label{prop:closure}
Languages recognised by HRAs are closed under union, intersection, concatenation and Kleene star.
\end{proposition}

\begin{proof}
We show concatenation and Kleene star only. For the former, consider HRAs $\AA_i=\langle Q_i,q_{0i},H_{0i},\delta_i,F_i\rangle$, $i=1,2$, and assume wlog that they have common type $(m,n)$. Let $w$ be an enlistment of all names in $H_{02}$ and construct $\AA_i'=\fix{\AA_i}{w}$, for $i=1,2$. Then, the concatenation $\LL(\AA_1)\circ\LL(\AA_2)$ is the language recognised by connecting $\AA_1'$ and $\AA_2'$ serially, that is, the automaton obtained by connecting each final state of $\AA_1'$ to the initial state of $\AA_2'$ with a transition labelled $[m{+}n]$, and with initial/final states those of $\AA_1'/\AA_2'$ respectively.

Finally, given an $(m,n)$-HRA $\AA$ and an enlistment $w$ of its initial names, we construct an automaton $\AA'$ by connecting the final states of $\fix{\AA}{w}$ to its initial state with a transition labelled $[m{+}n]$. We can see that $\LL(\AA')=\LL(\AA)^*$.
\end{proof}

As we shall next see, while universality is undecidable for HRAs,
their emptiness problem can be decided by reduction to coverability for transfer-reset vector addition systems. In combination, these results imply that HRAs cannot be effectively complemented. In fact, there are HRA-languages whose complements are not recognisable by HRAs. This can be shown via the following example, adapted from~\cite{Manuel_Ramanujam:2011}.

\begin{lemma}\label{lem:noncomplement}
HRAs are not closed under complementation.
\end{lemma}

\begin{example}\label{ex:complement}
Consider \ %the language
%\begin{align*}
$\LL_4 = \{w\in\names^*\ |\ \text{ not all names of $w$ occur exactly twice in it }\}$,
%\overline{\LL_4} &= \{w\in\names^*\ |\ \text{ all names in $w$ occur exactly twice in it }\}
%\end{align*}
which is accepted by the $(2,0)$-HRA below.
%, where ``$-$" can be any of $\empti,1,2$.%\vspace{-1.5mm}
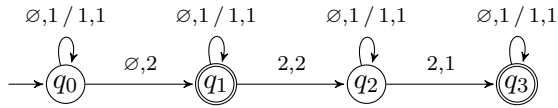
\begin{center}
\begin{tikzpicture}[automaton]
\node[state,initial] (q0) {$q_0$};
\node[state,accepting] (q1) [right=of q0] {$q_1$};
\node[state] (q2) [right=of q1] {$q_2$};
\node[state,accepting] (q3) [right=of q2] {$q_3$};
\def\loops{
  edge[loop above] node{$\scriptstyle \empti\ta1\,/\,1\ta 1$} ()
%  edge[loop below] node{$\scriptstyle\empti\ta1$} ()
}
\path[transition]
  (q0) \loops edge node[above]{$\scriptstyle\empti\ta2$} (q1)
  (q1) \loops edge node[above]{$\scriptstyle2\ta2$} (q2)
  (q2) \loops edge node[above]{$\scriptstyle2\ta1$} (q3)
  (q3) \loops;
\end{tikzpicture}%\vspace{-1.5mm}
\end{center}
The automaton non-deterministically selects an input name which either appears only once in the input or at least three times.
We claim that $\overline{\LL_4}$, the language of all words whose names occur exactly twice in them, is not HRA-recognisable. 
\end{example}

\begin{proof}
Suppose it were recognisable (wlog, \autoref{prop:regs_his}) by an $(m,0)$-HRA $\AA$ with $k$ states. Then, $\AA$ would accept the word
%\[
$w = a_1\ldots a_{k}\,a_1\ldots a_{k}$
%\]
where all $a_i$'s are distinct and do not appear in the initial assignment of $\AA$. 
Let $p=p_1p_2$ be the path in $\AA$ through which $w$ is accepted, with each $p_i$ corresponding to one of the two halves of $w$. Since all $a_i$s are fresh for $\AA$, the non-reset transitions of $p_1$ must carry labels of the form $(\empti,X)$, for some sets $X$.
Let $q$ be a state appearing twice in $p_1$, say $p_1=p_{11}(q)p_{12}(q)p_{13}$. 
Consider now the path $p'=p_1'p_2$ where $p_1'$ is the extension of $p_1$ which repeats $p_{12}$, that is, $p_1'=p_{11}(q)p_{12}(q)p_{12}(q)p_{13}$. We claim that $p'$ is an accepting path in $\AA$. Indeed, by our previous observation on the labels of $p_1$, the path $p_1'$ does not block, i.e.~it cannot reach a transition $q_1\xrightarrow{X\ta Y}q_2$, with $X\not=\empti$, in some configuration $(q_1,H_1)$ such that $H_1\at X=\empti$. 
We need to show that $p_2$ does not block either (in $p'$).
Let us denote $(q,H_1)$ and $(q,H_2)$ the configurations in each of the two visits of $q$ in the run of $p$ on $w$; and let us write $(q,H_3)$ for the third visit in the run of $p_1'$, given that for the other two visits we assume the same configurations as in $p$. Now observe that, for each nonempty $X\subseteq[m]$, repeating $p_{12}$ cannot reduce the number of names appearing precisely in $X$, therefore $|H_2\at X|\leq|H_3\at X|$. The latter implies that, since $p$ does not block, $p'$ does not block either. Now observe that any word accepted by $w'$ is not in $\overline{\LL_4}$, as $p_1'$ accepts more than $k$ distinct names, a contradiction.
\end{proof}

\section{Removing Registers}\label{sec:registers}

Although registers are convenient for expressing some languages
  (\autoref{ex:HRA-regs-no} and \autoref{ex:HRA-regs-yes}),
 when reasoning about HRAs it is more convenient to focus on histories only.
In this section, we show that this approach is sound and in particular we present three ways of removing registers:
\begin{itemize}
\item we can construct a bisimilar automaton if we are allowed to use extra histories and resets;
\item we can preserve language if we are allowed to use extra histories (but no resets);
\item we can preserve emptiness if we are allowed to use extra states
  (but no histories nor resets).
\end{itemize}
Each of these constructions will be useful in the sequel either for devising emptiness checks and, more generally, they demonstrate that registers can be considered as a derivative notion.

\subsection{Simulating Registers with Histories and Resets} % <<<

The semantics of registers is very similar to that of histories.
The main difference is that registers are forced to contain at most one name.
To simulate registers with histories, we reset histories before inserting names.
Resetting histories might cause the automaton to forget names
  that are necessary for deciding how to proceed.
The solution is to use two histories for each register:
  one holds the old name, the other holds the new name.
Only the history where the new name will be written needs to be reset.

\begin{proposition}\label{prop:regs_his}
Let $\AA=\langle Q,q_0,H_0,\delta,F\rangle$ be an $(m,n)$-HRA.
We can construct an $(m+2n,0)$-HRA $\AA'=\langle Q',q'_0,H'_0,\delta',F'\rangle$
  that is bisimilar to~$\AA$.
We have $|Q'|\in O(2^n|Q|)$ and $|\delta'|\in O(2^n|\delta|)$.
The construction can be done in $O\bigl((m+n)(|Q'|+|\delta'|)\bigr)$ time.
\end{proposition}

\begin{proof}
For each $q$ in~$Q$, we include $2^n$~states $(q,f)$ in~$Q'$,
  where $f:[n]\to[2n]$ is such that $f(i)\in\{i,i+n\}$ for all~$i$.
The name that $\AA$~holds in register~$i$ will be found in history~$m+f(i)$ of~$\AA'$.
We set $\bar{f}$ to be the complement of $f$;
  that is, $\bar{f}(i) \defeq n+2i-f(i)$.
Let $f^\dagger(i)$ be $i$ if $i\in[m]$ and $m+f(i)$ otherwise.
Moreover, $q_0'=(q_0,\id)$,\quad $F'=\{(q,f)\ |\ q\in F\}$, and $H_0'$ is $H_0$ extended so that $H_0'(i)=\empti$ for all $i>m+n$.
Finally, we include in $\delta'$ precisely the following transitions.
\begin{itemize}
\item For each $q\trdelta{X,X'}q'\in\delta$, add \hbox{$(q,f)\trdelta{Y}\cdot\xtrdelta{f^\dagger(X),\bar{f}^\dagger(X')}(q',f')$} where $Y=[m{+}1,m{+}2n]\setminus \range{f}$, and $f'$ is given by: $f'(i)=\bar{f}(i)$ if $i\in X\cup X'$ 
and $f'(i)=f(i)$ otherwise.
(Note that we need a few extra states in $Q'$, which remain nameless in this proof.)
\item For each $q\trdelta{X}q'\in\delta$, add $(q,f)\xtrdelta{f^\dagger(X)}(q',f)$.
\end{itemize}
The relation $\{\, ((q,H),((q,f),H'))\mid H = H'\circ f^\dagger\,\}$ witnesses bisimilarity.
\end{proof}

% >>>
\subsection{Replacing Registers with Histories using Colours} % <<<

The result of \autoref{prop:regs_his} comes at the cost of introducing reset transitions
  even if the original automaton does not have such transitions.
Reset transitions are undesirable because they increase the complexity of the emptiness problem
  (\autoref{sec:empty}).
We can avoid introducing reset transitions by using the \emph{colouring technique} of~\cite{CMA}.
The construction is more involved and the resulting automaton is not bisimilar to the original,
  but it is language equivalent.

Before we proceed with the proof, let us illustrate the technique on two examples.
\autoref{ex:HRA-regs-no} illustrates how to check that a name is not in the simulated register;
\autoref{ex:HRA-regs-yes} illustrates also how to check that a name is in the simulated register.

\begin{example}\label{ex:HRA-regs-no}
The language
\begin{align*}
\LL_5 \;=\; \{\,a_1\ldots a_n\mid\text{$a_i \ne a_{i+1}$ for all $i$}\,\}
\end{align*}
is recognized by both of the automata below (with histories initially empty):
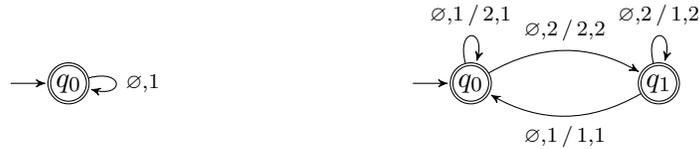
\begin{center}
\begin{tikzpicture}[automaton,baseline=0pt]
\node[state,initial,accepting] (q00) {$q_0$};
\path[transition]
  (q00) edge[loop right] node{$\scriptstyle \empti,1$} ();
\end{tikzpicture}
\hfil
\begin{tikzpicture}[automaton,baseline=0pt]
\node[state,initial,accepting] (q0) {$q_0$};
\node[state,accepting] (q1) [right=25mm of q0] {$q_1$};
\path[transition]
  (q0) edge[bend left] node[above]{$\scriptstyle\empti\ta2\,/\, 2,2$} (q1)
  (q1) edge[bend left] node[below]{$\scriptstyle\empti,1\,/\, 1,1$} (q0)
  (q0) edge[loop above] node{$\scriptstyle \empti,1\,/\,2,1$} ()
  (q1) edge[loop above] node{$\scriptstyle \empti,2\,/\,1,2$} ();
\end{tikzpicture}
\end{center}
The one on the left is a $(0,1)$-HRA, and we can straight away see its accepted language is~$\LL_5$.
The one on the right is a $(2,0)$-HRA for which it is less clear why it accepts~$\LL_5$.
The reason is the following invariant:
\begin{itemize}
\item $H(1) \uplus H(2)$ is a partition of all names seen so far; and
\item if the state is $q_k$ then the last seen name is in $H(k+1)$, for $k\in\{0,1\}$; and
\item all possible partitions of the names are realisable.
\end{itemize}
The first two points are easy to check.
Because all transitions have labels of the form $(X,\{i\})$,
  all seen names are remembered in precisely one history.
Because all transitions incoming into $q_k$ have labels the form $(X,\{k+1\})$,
  the last seen name is remembered in $H(k+1)$.
The consequence of these first two points is
  that the automaton on the right accepts a name only if it is different from the last seen name.
Indeed, all transitions outgoing from $q_k$ have labels of the form $(X,X')$ with $k+1 \notin X$.
Thus, the first two points should be seen as lemmas
  which let us establish that all words accepted by the automaton on the right belong to~$\LL_5$.

Informally, the third point is the key lemma that lets us establish the converse,
  that all words in~$\LL_5$ are accepted.
However, to see why this is so, we need to rephrase it in a more formal way.
Given a word $a_1\ldots a_n \in \LL_5$,
  we consider an arbitrary partition $H(1) \uplus H(2)$ of the set $\{a_1,\ldots,a_n\}$.
Without loss of generality, assume $a_n \in H(1)$.
The claim is that there exists a run of the automaton on the right
  that accepts the word $a_1\ldots a_n$ and ends in configuration $(q_0,H)$.
We can prove this by induction.
If $a_{n-1} \in H(k+1)$ then the previous state in the run must have been~$q_k$,
  for $k \in \{0,1\}$.
Suppose $a_{n-1} \in H(2)$; the other case is symmetric.
Then, by the induction hypothesis,
  we know that there is a run that accepts the word $a_1\ldots a_{n-1}$
    and ends in configuration $(q_1, H')$,
  where we choose
\begin{align*}
  H'(1) &\defeq
  \begin{cases}
  H(1) \setminus \{a_n\} &\text{if $a_n \notin \{a_1,\ldots,a_{n-1}\}$} \\
  H(1) \cup \{a_n\} &\text{if $a_n \in \{a_1,\ldots,a_{n-1}\}$}
  \end{cases}
&
  H'(2) &\defeq H(2) \setminus \{a_n\}
\end{align*}
It is clear that $H'(1) \uplus H'(2)$ is a partition of $\{a_1,\ldots,a_{n-1}\}$, as required.
Finally,
  we need to show that $(q_1,H') \trdelta{a_n} (q_0,H)$ belongs to the configuration graph.
If $a_n \notin \{a_1,\ldots,a_{n-1}\}$,
  then this is true because of $q_1 \trdelta{\empti,1} q_0$;
if $a_n \in \{a_1,\ldots, a_{n-1}\}$,
  then this is true because of $q_1 \trdelta{1,1} q_0$.
\end{example}

\begin{example}\label{ex:HRA-regs-yes}
The language
\begin{align*}
\LL_6 \;=\; \{\,a_1\ldots a_n\mid\text{$a_i = a_{i+1}$ iff $i$ is odd}\,\}
\end{align*}
is recognized by both of the following automata:
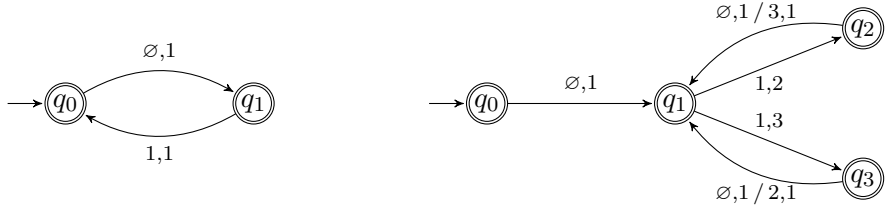
\begin{center}
\begin{tikzpicture}[automaton,baseline=0pt]
\node[state,initial,accepting] (q0) {$q_0$};
\node[state,accepting] (q1) [right=25mm of q0] {$q_1$};
\path[transition]
  (q0) edge[bend left] node[above]{$\scriptstyle \empti,1$} (q1)
  (q1) edge[bend left] node[below]{$\scriptstyle 1,1$} (q0);
\end{tikzpicture}
\hfil
\begin{tikzpicture}[automaton,baseline=0pt]
\node[state,initial,accepting] (q0) {$q_0$};
\node[state,accepting] (q1) [right=25mm of q0] {$q_1$};
\node (0) [right=25mm of q1] {};
\node[state,accepting] (q2) [above=10mm of 0] {$q_2$};
\node[state,accepting] (q3) [below=10mm of 0] {$q_3$};
\path[transition]
  (q0) edge node[above]{$\scriptstyle \empti,1$} (q1)
  (q1) edge node[below]{$\scriptstyle 1,2$} (q2)
  (q2) edge[bend right] node[above]{$\scriptstyle \empti,1 \,/\, 3,1$} (q1)
  (q1) edge node[above]{$\scriptstyle 1,3$} (q3)
  (q3) edge[bend left] node[below]{$\scriptstyle \empti,1 \,/\, 2,1$} (q1)
;
\end{tikzpicture}
\end{center}
The one on the left is a $(0,1)$-HRA, while the one on the 
right is a $(3,0)$-HRA. As in the previous example, the fact that the $(3,0)$-HRA accepts~$\LL_6$ is not immediately clear.
%\nt{I rewrote these sentences as I tend to dislike repetition in prose.}
Informally, the reason is the following invariant:
\begin{itemize}
\item $H(1) \uplus H(2) \uplus H(3)$ is a partition of the names seen so far;
\item if the state is $q_k$ then the last seen name is in $H(k)$, for $k\in\{1,2,3\}$;
\item $|H(1)|=1$ in $q_1$, and $|H(1)|=0$ otherwise; and
\item
  for each partition $H(1) \uplus H(2) \uplus H(3)$
    that is compatible with the previous constraints,
  there is a nondeterministic run that realises it.
%\nt{I'm not sure this invariant holds. it seems to suggest that we can put names $H(2)$ and $H(3)$ however we want. Also, some text after the itemize environment would be nice.}
\end{itemize}
The first two points hold for the same reasons as in \autoref{ex:HRA-regs-no}.
The third point holds because
  all incoming transitions of~$q_1$ insert a name in $H(1)$,
  all outgoing transitions of~$q_1$ remove a name from $H(1)$,
  and all runs alternate between state $q_1$ and some other state.
After an odd number of names was processed,
  the automaton is in state~$q_1$ and $H(1)$ contains only the last name seen.
All outgoing transitions from~$q_1$ accept a name only if it is in $H(1)$\,---\,%
  in other words, if it equals the last seen name.
After an even and positive number of names was processed,
  the automaton is in state $q_2$~or~$q_3$.
All outgoing transitions from~$q_2$ accept a name only if it is \emph{not} in $H(2)$,
  where the last seen name is;
  $q_3$ acts symmetrically.
Thus, the first three points let us establish
  that all words accepted by the automaton on the right belong to~$\LL_6$.

As in \autoref{ex:HRA-regs-no},
  the last point lets us establish that all word in $\LL_6$ are accepted.
Given that the proof of the last point is very similar to the one in \autoref{ex:HRA-regs-no},
  let us only sketch it.
Given a word $a_1\ldots a_n \in \LL_6$ with $n>0$,
  we consider an arbitrary partition $H(1) \uplus H(2) \uplus H(3)$
    of the set $\{a_1,\ldots,a_n\}$
  such that $H(1) \subseteq\{a_n\}$.
Let $k$ be such that $a_n \in H(k)$.
Formally, the claim of the last point is that there exists a run of the automaton on the right
  that is labelled by the word $a_1\ldots a_n$ and ends in the configuration $(q_k,H)$.
The case $n$ odd and~$>1$ is similar to the previous example:
We pick $k'$ such that $a_{n-1} \in H(k')$ and
\begin{align*}
  H'(1) &\defeq \emptyset
&
  H'(k') &\defeq H(k') \setminus \{a_n\}
&
  H'(5-k') &\defeq
  \begin{cases}
  H(5-k') \setminus \{a_n\}\! &\!\text{if $a_n \notin \{a_1,\ldots,a_{n-1}\}$} \\
  H(5-k') \cup \{a_n\}\! &\!\text{if $a_n \in \{a_1,\ldots,a_{n-1}\}$}
  \end{cases}
\end{align*}
Then we invoke the induction hypothesis to show there is a run that accepts $a_1\ldots a_{n-1}$
  and ends in configuration $(q_{k'},H')$.
In the case $n$~even, we pick
\begin{align*}
  H'(1) &\defeq \{a_n\}
&
  H'(2) &\defeq H(2) \setminus \{a_n\}
&
  H'(3) &\defeq H(3) \setminus \{a_n\}
\end{align*}
and then invoke the induction hypothesis to show that there is a run
  that accepts $a_1\ldots a_{n-1}$ and ends in configuration $(q_1,H')$.
We skip the case $n=1$ in this proof sketch.
\end{example}

\medskip

We are now ready for the general result, which we prove in two steps.
The main construction will be presented first and only applies to HRAs with initially empty registers.
(The correctness of the main construction requires that certain graphs are $2$-colourable.
The arguments in the previous two examples can be seen as giving explicit colouring algorithms
  that work in special cases.)
At a second stage, we show how to initially simulate nonempty registers at the expense of some more additional histories.

\begin{proposition}\label{prop:nrHRA-regs}
Let $\AA=\langle Q,q_0,H_0,\delta,F\rangle$ be an $(m,n)$-non-reset-HRA
  with registers initially empty.
We can construct an $(m+3n,0)$-non-reset-HRA $\AA'=\langle Q',q'_0,H'_0,\delta',F'\rangle$
  that accepts the same language as~$\AA$.
We have $|Q'|\in O(2^{2n}\cdot|Q|)$ and $|\delta'|\in O(2^{3.3n}\cdot|\delta|)$.
%The construction can be done in $O\bigl((m+n)|\delta'|\bigr)$ time.
%\nt{Could we have the bounds in the form $2^x$, instead of $x^n$? Also, I am not sure the time bound is necessary here.}
% RG: The time bound is not needed. I'll leave it commented for a little while.
\end{proposition}

\begin{proof}
{Each register~$i$ will be simulated by three histories,
  named $i_{\rm R}$, $i_{\rm B}$ and $i_{\rm Y}$ respectively.%
\footnote{%
  B and Y are the black and yellow colours of~\cite{CMA};
  R stands for `read'.
}
Each state $q \in Q$ will be simulated by several states $(q,f) \in Q'$,
  where $f : [m+1,m+n] \to \{\varnothing,{\rm R}, {\rm B}, {\rm Y}\}$.
The construction will ensure the following invariant:
\begin{itemize}
\item $H(i_{\rm R}) \uplus H(i_{\rm B}) \uplus H(i_{\rm Y})$
is a partition of the names that have been written to register~$i$ and have subsequently been rewritten by other names or are still in register~$i$;\footnote{note that a name can also be transferred out of register $i$ (via transition with label $X,Y$ where $i\in X\setminus Y$), instead of being directly rewritten, in which case we would not store it in
$H(i_{\rm R}) \uplus H(i_{\rm B}) \uplus H(i_{\rm Y})$.}
\item $|H(i_{\rm R})| = 1$ if $f(i)={\rm R}$, and $|H(i_{\rm R})|=0$ otherwise; and
\item if $f(i)\not=\varnothing$ then the current name of register~$i$ is in $H(i_{f(i)})$; register $i$ is empty otherwise.
\end{itemize}
Thus, according to the last point above, $f$ records in which of histories $i_{\rm R},i_{\rm B},i_{\rm Y}$ has the current name of register $i$ been stored. We shall instrument $\AA'$ in such a way that it will only store that name, say $a$, in $i_{\rm R}$ if the next time register $i$ is being invoked by $\AA$ is for reading $a$. 
This way we shall ensure that $i_{\rm R}$ never contains more than one name. 
Otherwise, i.e.\ if $\AA$ next invokes $i$ for overwriting its contents, $f$ will be mapping $i$ to one of $i_{\rm B},i_{\rm Y}$. These can be seen as garbage collecting histories: they contain all names that have passed from register $i$ and will not be immediately read from it. The reason why we need two of these, $i_{\rm B}$ and $i_{\rm Y}$, is to be able to reuse old names of register $i$ without running the risk of confusing them with its current name $a$.}

To simulate one transition $q \trdelta{X,X'} q'$, {accepting say a name $a$,}
  we shall use several transitions of the form $(q,f) \trdelta{Z,Z'} (q',f')$
where 
$f$~and~$f'$ agree {outside~$X\cup X'$.}
Let us consider an arbitrary such pair $(f,f')$, and see how to pick $Z$~and~$Z'$.
On histories, $X$~and~$Z$ coincide: $X \cap [m] = Z \cap [m]$.
For each register $i \in X \setminus [m]$,
  we need that {$a$} be equal to the name currently written in register~$i$.
But, we know the current name in register~$i$ only if $f(i)={\rm R}$.
That is why we include a transition $(q,f)\trdelta{Z,Z'}(q',f')$
  only if {$X \setminus [m]\subseteq f^{-1}({\rm R})$},
  and we include $\{\,i_{\rm R}\mid i \in X \setminus [m]\,\}$ in~$Z$.
For each register $i \in [m+1,m+n] \setminus X$,
  we must ensure that {$a$} is not equal to the current name in register~$i$, which resides in $H(i_{f(i)})$.
Hence, we pick all $Z$ such that
\[
  Z = (X \cap [m]) \uplus Z_1 \uplus Z_0
\]
where
 $ Z_1 = \{\,i_{\rm R}\mid i\in X \setminus [m]\,\}$ and 
$ Z_0 \subseteq \{\,i_{x}\mid i \in [m+1,m+n] \setminus X\land x\in\{{\rm B,Y}\}\land x\not=f(i)\,\}$ is such that, for all $i$, $|Z_0\cap\{i_{\rm B},i_{\rm Y}\}|\leq1$.
Now we must write the current name to histories $X' \cap [m]$,
  and we must simulate writing the current name to registers $X' \setminus [m]$.
For each $i \in X' \setminus [m]$ we shall \emph{nondeterministically} write the current name
  to one of $i_{\rm R}, i_{\rm B}, i_{\rm Y}$ by \emph{guessing} whether register $i$ will be used next for reading or not.
The place where we write is given by $f'(i)$,
that is,
\begin{align*}
  Z' = (X' \cap [m]) \uplus \{\,i_{f'(i)}\mid i\in X' \setminus [m]\,\}.
\end{align*}
Finally, we make sure that all $i\in ([m+1,m+n]\cap X)\setminus X'$ in $f'$ are mapped to $\varnothing$, as these registers are now empty, i.e.\ we impose $f'(([m+1,m+n]\cap X)\setminus X')\subseteq\{\empti\}$.

Thus, in summary, we take $Q'=Q\times([m+1,m+n]\to\{\varnothing,{\rm R,B,Y}\})$ and:
\begin{itemize}
\item $q_0'=(q_0,\{(i,\varnothing)\mid i\in[m+1,m+n]\})$ and $F'=F\times([m+1,m+n]\to\{\varnothing,{\rm R,B,Y}\})$;
\item $H_0'=H_0\cup\{(i,\empti)\mid i\in[m+n+1,m+3n]\}$;
\item we include $(q,f)\trdelta{Z,Z'}(q',f')$ in $\delta'$ just if there is some $q\trdelta{X,X'}q'$ in $\delta$ such that:
\begin{itemize}
\item $X\setminus[m]\subseteq f^{-1}(\rm R)$, 
\item for all $i\in[m+1,m+n]\setminus(X\cup X')$, $f'(i)=f(i)$,
\item for all $i\in ([m+1,m+n]\cap X)\setminus X'$, $f'(i)=\varnothing$,
\item  $Z = (X \cap [m])\uplus \{i_{\rm R}\mid i\in X \setminus
  [m]\}\uplus Z_0$ \\\phantom{$Z={}$} with $ Z_0 \subseteq \{i_{x}\mid i \in [m+1,m+n] \setminus X \land x\not=f(i)\}$,
\item $Z' = (X' \cap [m]) \uplus \{\,i_{f'(i)}\mid i\in X' \setminus [m]\,\}$.
\end{itemize}
\end{itemize}
\cutout{
(Observe that we nondeterministically guess $f^{-1}({\rm R})$.
If we guess right the simulation continues; if we guess wrong the simulation gets stuck.
Similarly, by nondeterministically writing in one of $i_{\rm B}, i_{\rm Y}$,
  we guess which old name was overwritten and is not the in register~$i$ any more,
  to implement disequality checks later on.)
}%
Let us now see what is the size of the HRA~$\AA'$ so constructed.
We have $|Q'|\in O(4^n \cdot |Q|)$.
For each transition $q \trdelta{X,X'} q'$ in~$\AA$,
  we introduce several transitions $(q,f) \trdelta{Z,Z'} (q'f')$ in~$\AA'$.
Let us count how many.
There are $\le 4^n$ choices for~$f$;
  there are $\le 3^n$ choices for~$Z_0$
    because for each $i$ we pick~$i_{\rm B}$, or~$i_{\rm Y}$, or none of them; and
  there are $\le 3^n$ choices for~$f'$
    because $f'(X')\subseteq\{i_{\rm R},i_{\rm B},i_{\rm Y}\}$
    and $f'$~is uniquely determined outside~$X'$.
In summary, $|\delta'|\in O(2^{2(1+\log_2 3)n} \cdot |\delta|)$.

{Finally, we show that $\LL(\AA)=\LL(\AA')$. 
Let first $w\in\LL(\AA')$ have an accepting transition path $p'$ in $\AA'$ with edges $(q_k,f_k)\trdelta{Z_k,Z_k'}(q_{k{+}1},f_{k{+}1})$, for $k=1,\ldots,N$.
Reading the definition of $\delta'$ backwards, this yields an accepting transition path $p$ in $\AA$ with edges $q_k\trdelta{X_k,X_k'}q_{k{+}1}$ where
\begin{itemize}
\item $X_k=(Z_k\cap[m])\cup\{\,i\mid i_{\rm R}\in Z_k\,\}$,
\item $X_k'=(Z_k'\cap[m])\cup\{\,i\mid \{i_{\rm R},i_{\rm B},i_{\rm Y}\}\cap Z_k'\not=\emptyset\,\}$.
\end{itemize}
To see that $p$ accepts $w$, suppose that $p'$ yields a sequence of configurations $((q_k,f_k),H_k')$.
Then, by induction, we can show that $p$ yields a sequence of configurations $(q_k,H_k)$, where:
\begin{itemize}
\item for all $i\in[m]$, $H_k(i)=H_k'(i)$; 
\item for all $i\in[m+1,m+n]$, if $f_k(i)\not=\varnothing$ then $H_k(i)=\{a\}$ for some $a\in H_k'(i_{f_k(i)})$, otherwise $H_k(i)=\empti$;
\item for all names $a$, if $a\in H_k'@Z_k$ then $a\in H_k@X_k$.
\cutout{for all $x\in X_k$:
\begin{itemize}
\item if $x=i_{\rm R}$ then $H_k(i_{\rm R})=H_k'(i)=\{a\}$ and $i\in X_k$;
\item if $x=i_{\rm B}$ then $f_k(i)=\rm Y$ and therefore $a\notin H_k'(i)$, and also $i\notin X_k$; dually of $x=i_{\rm Y}$;
\item if $x\in[m]$ then $a\in H_k'(x)$.
\end{itemize}}
\end{itemize}
Hence, $w\in\LL(\AA)$.
\\
Conversely, let $w=a_1\cdots a_{N}\in\LL(\AA)$ have an accepting transition path $p$ in $\AA$ with edges $q_k\trdelta{X_k,X_k'}q_{k{+}1}$ for $k=1,\ldots,N$. We construct a corresponding accepting path $p'$ in $\AA'$ with edges $(q_k,f_k)\trdelta{Z_k,Z_k'}(q_{k{+}1},f_{k{+}1})$ as follows. We have that $f_0=\{(i,\empti)\ |\ i\in[m]\}$. Moreover, $Z_k=(X_k\cap[m])\cup W_k$ and 
$Z_k'=(X_k'\cap[m])\cup W_k'$ where:
\begin{enumerate}[label=\({\alph*}]
%\item For each position $k$ such that in all previous appearances of $a_k$ in $w$, say as $a_k=a_{k'}$ with $k'<k$ (if any), $a_{k'}$ is not stored in registers (i.e.~$X_{k'}'\subseteq[m]$), we set $W_k=\empti$. 
\item For each position $k$ such that the previous appearance of $a_k$ in $w$ is some $a_{k'}=a_k$ with $k'<k$, we set $W_k=W_{k'}'$. If there is no previous appearance, we set $W_k=\empti$.
\item For each position $k$ and $i\in X_k'\setminus[m]$ such that the next appearance of $a_k$ in $w$ is some $a_{k'}$ with $i\in X_{k'}$, we include $i_{\rm R}$ in $W_k'$.
\item For each position $k$ and $i\in X_k'\setminus[m]$ such that the next appearance of $a_k$ in $w$ is some $a_{k'}$ with $i\notin X_{k'}$, we include in $W_k'$ one of $i_{\rm B},i_{\rm Y}$. We do the same also if there is no next appearance of $a_k$ in $w$.
\end{enumerate}
The above specifications determine the values of all $W_k,W_k'$, modulo the choice between ${\rm B}$ and $\rm Y$ in case (c). 
Clearly, if the path $p'$ can be so constructed then $\AA'$ accepts $w$. It remains to show that $p'$ can indeed be implemented in $\AA'$.
The form of the $f_k$'s is derived from (a-c) according to the definition of $\delta'$. But note that the definition of $\delta'$ imposes the following condition:
\begin{enumerate}[label=\({\alph*}]
\item[\(d] For each position $k$ and $i_{\rm Y}\in W_k'$ such that the next appearance of any $i_x$ in $p'$ is in some $W_{k'}$ (with $k<k'$), we must have $i_{\rm B}\in W_{k'}$. Dually if $i_{\rm B}\in W_k'$.
\end{enumerate}
For example,
  if $i_{\rm Y} \in W'_2$
    but none of $i_{\rm B}, i_{\rm Y}, i_{\rm R}$ occurs in any of $W_3,W'_3,W_4,W'_4,W_5,W'_5$,
  then $\{i_{\rm B},i_{\rm Y},i_{\rm R}\}\cap W_6\subseteq\{i_{\rm B}\}$.
This condition stems from the interdiction to include $i_{f(i)}$ in $W_{k'}$ when $f(i)\not=i_{\rm R}$.
The consequence of condition~(d) is that in case~(c) above
  we cannot pick ${\rm B}$~and~${\rm Y}$ arbitrarily.
%Thus, condition (d) restricts the ways in which we can pick $\rm B$ and $\rm Y$ in case (c) above. 

We need to show that a choice of ``colours'' ($\rm B$ and $\rm Y$) satisfying both (c) and (d) can be made. We achieve this by applying a graph colouring argument.
Let us define a labelled graph $\mathcal{G}$ with:
\begin{itemize}
\item Vertices $(k,i)$ and $(k,i)'$ for each $k\in[0,N-1]$ and $i\in[m+1,m+n]$;
\item For each $i$ and $k<k'$ as in (c) above, an edge between $(k,i)'$ and $(k',i)$ labelled with ``$=$''.
\item For each $i$ and $k<k'$ as in (d) above, an edge between $(k,i)'$ and $(k',i)$ labelled with ``$\not=$''.
\end{itemize}  
Then, a valid choice of colours can be made as long as $\mathcal{G}$ can be coloured with $\rm B$ and $\rm Y$ in such a way that $=$-connected vertices have matching colours, while $\not=$-connected vertices have different colours. For the latter, it suffices to show that the graph obtained by merging $=$-connected vertices can be 2-coloured, for which it is enough to show that $\mathcal{G}$ contains no cycles. Suppose $\mathcal{G}$ contained a cycle. Then, by definition of the edge relation of the graph, it must be the case that the leftmost vertex (i.e.\ the one with the least $k$ index) in the cycle be some $(k,i)'$.
The vertex $(k,i)'$ has two outgoing edges, one for each label. The $\not=$-edge in particular connects to some $(k',i)$ such that $k<k'$, obtained from condition (d). Since $(k',i)$ is part of the cycle, it must have an outgoing $=$-edge to some vertex $(k'',i)'$ with $k''<k'$. But note that condition (d) stipulates that there is no mention of $i$ between $k$ and $k'$ in $p'$, and therefore $k''\leq k$. Moreover, $k''=k$ is not an option as it would imply that register $i$ was not rewritten between steps $k$ and $k'$ in $p$, in which case $k$ and $k'$ would fall under case (b) above. Hence, $k''<k$ which contradicts our assumption that $(k,i)'$ was the leftmost vertex in the cycle.}
\cutout{
our above labelling is correct and, in particular, that there is a valid choice of labels from $\{i_b,i_y\}$ in step (c) above. By a valid choice we mean one such that if, for instance, $Y_k'=i_b$ then $f_{k'}(i)\not=b$ and therefore $X_{k'}'$ can correctly pick $a_k$. By definition, the value of $f_{k'}(i)$ is determined by the rightmost position $l$, with $k<l<k'$, which assigns a name in either of $i_r,i_b,i_y$ (since $a_k$ becomes locally fresh from state $q_k$ to $q_{k'}$, such an $l$ always exists). Our constraint, therefore, (in this case) is that $Y_l\not=b$. In particular, if position $l$ falls under case (c) above, we need to choose $Y_l=y$. Thus, it suffices to colour all our positions which fall under case (c) with colours from $b,y$, such that no inter-related $q_k$ and $q_l$ have the same colour. We claim that such a colouring is always possible. Indeed, we can build a graph $\mathcal{G}$ as follows. The nodes of $\mathcal{G}$ are the positions which fall under case (c), arranged on a line in ascending order, from left to right. Now, for each inter-related positions $k$ and $l$ as above, we add a directed edge from node $k$ to node $l$. Then, a valid colouring of our positions is possible iff $\mathcal{G}$ is 2-colourable, i.e.~it contains no cycles. Now note that each node in $\mathcal{G}$ has at most one outgoing edge, and all edges go from left to right (above, $k<l$). Hence, $\mathcal{G}$ is 2-colourable and we can therefore build a valid path $p'$ in $\AA'$ accepting $w$.
}
\end{proof}

\cutout{The construction in the proof above could be extended
  such that $\AA'$ contains reset transitions if and only if $\AA$ contains reset transitions.
However,
  instead of complicating the construction from above,
  we can simply use \autoref{prop:regs_his}
    to handle the case in which $\AA$ contains reset transitions.}

Note that the above result can be extended
  to handle the case in which registers are not initially empty by simply 
 making use of \autoref{lem:forclos}. However, the construction in that lemma leads to a doubly exponential blow-up in size, which we can be avoided by the alternative approach that follows.

\begin{proposition}\label{prop:empty-init-histories}
Let $\AA=\langle Q,q_0,H_0,\delta,F\rangle$ be an $(m,n)$-non-reset-HRA.
We can construct a bisimilar $(m+n,n)$-non-reset-HRA $\AA'=\langle Q',q'_0,H'_0,\delta',F'\rangle$
such that, for all $i\in[m+n+1,m+2n]$, $H_0(i)=\empti$.
Moreover, we have $|Q'|\in O(2^n\cdot|Q|)$ and $|\delta'|\in O(2^{2n}\cdot|\delta|)$.
%The construction can be done in $O\bigl((m+n)|\delta'|\bigr)$ time.
\end{proposition}
\begin{proof}
The main idea behind the construction of $\AA'$ is to use the additional $n$ histories to store just the initial names of the registers in $\AA$. Once these names have been used in the computation, they are transferred to their actual registers (if any). We will also need to track which of the registers in $\AA$ are still simulated by histories in $\AA'$. Thus, we set \[
Q' = Q\times([m+1,m+n]\to\{0,1\})
\]
and $q_0'=(q_0,\{(i,1)\mid i\in[m+1,m+n]\})$, $H_0'=H_0\cup\{(m+n+i,\empti)\mid i\in[n]\}$ and $F'=F\times([m+1,m+n]\to\{0,1\})$. Moreover, for each $q\trdelta{X,X'}q'$ in $\delta$ and map $f$, we include in $\delta'$ a transition $(q,f)\trdelta{Y,Y'}(q',f')$ where:
\begin{itemize}
\item $Y=(X\cap[m])\uplus Y_1\uplus Y_0$, where 
\[
Y_1=\{\,i\in X\mid f(i)=1\,\}\cup
\{\,n+i\mid i\in X\land f(i)=0\,\}
\]
and $Y_0\subseteq \{\,i\in[m+1,m+n]\mid i\notin X\land f(i)=0\,\}$;
\item $Y'=(X'\cap[m])\cup\{\,n+i\mid i\in X'\setminus[m]\,\}$;
\item $f'=f[i\mapsto 0\mid i\in X\cup X']$.
\end{itemize}
Then, taking $R$ to be the relation:
\begin{align*}
R=\{\,
((q,H),((q,f),H')\mid&\;
H\upharpoonright[m]=H'\upharpoonright[m]
\land \forall i\in[m+1,m+n].\\
&\land f(i)=1\implies H'(n+i)=\empti\land H'(i)=H(i)\\
&\land f(i)=0\implies H'(n+i)=H(i)\\
&\land \forall j\in[m+1,m+n].\,H'(i)\cap H'(n+j)=\empti
\,\}
\end{align*}
we can show that $R$ is a bisimulation.

Let us now see what is the size of the HRA~$\AA'$ we constructed.
We have $|Q'|\in O(2^n \cdot |Q|)$.
For each transition $q \trdelta{X,X'} q'$ in~$\AA$,
  we introduce several transitions $(q,f) \trdelta{Y,Y'} (q'f')$ in~$\AA'$:
there are $\le 2^n$ choices for~$f$;
  and there are $\le 2^n$ choices for~$Y_0$.
In summary, $|\delta'|\in O(2^{2n} \cdot |\delta|)$.
\end{proof}

Hence, the general case follows.

\begin{corollary}\label{cor:nrHRA-regs}
Let $\AA=\langle Q,q_0,H_0,\delta,F\rangle$ be an $(m,n)$-non-reset-HRA.
We can construct an $(m+4n,0)$-non-reset-HRA $\AA'=\langle Q',q'_0,H'_0,\delta',F'\rangle$
  that accepts the same language as~$\AA$.
We have $|Q'|\in O(2^{3n}\cdot|Q|)$ and $|\delta'|\in O(2^{5.3n}\cdot|\delta|)$.
%\nt{Need to replace $X$ with the bound from previous proposition}
%The construction can be done in $O\bigl((m+n)|\delta'|\bigr)$ time.
\end{corollary}

% >>>
\subsection{Simulating Registers Symbolically} % <<<

So far we saw how to simulate registers using histories.
If we are interested only in emptiness/reachability rather than language equivalence,
we can actually simulate the behaviour of registers without the inclusion of additional histories. 
This alternative is going to be crucial in \autoref{sec:unary}, where the number of histories will be fixed to just one.

We next describe how this simulation can be done. Given an assignment $H$ with $m$ histories and $n$ registers, we can represent $H$ \emph{symbolically} as follows:
\begin{itemize}
\item we map each name stored in the registers of $H$ to a number from the set $[n]$;
\item we subsequently replace in $H$ all these names by their number.
\end{itemize}
For example,
consider the assignment
\begin{align*}
    \bigl\{
      1\mapsto\{a,b,c\},\,
      2\mapsto\{d\},\,
      3\mapsto\empti,\,
      4\mapsto\{a\},\,
      5\mapsto\{d\}
    \bigr\}
\end{align*}
of a $(1,4)$-HRA\null.
We can simulate it symbolically by mapping $d$ to $1$, and $a$ to $2$. This results to a symbolic representation:
\begin{align*}
      \bigl\{
        1\mapsto\{2,b,c\},\,
        2\mapsto 1,\,
        3\mapsto\empti,\,
        4\mapsto 2,\,
        5\mapsto 1
      \bigr\}
\end{align*}
where the nominal part has been curtailed to the fact that $H(1)$ contains the names $b$ and $c$.

We can now employ this representation technique to represent configurations of $(m,n)$-HRAs by corresponding ones belonging to $(m,0)$-HRAs.
In particular, given the configuration
\begin{align*}
  \Bigl(
    q,\;
    \bigl\{
      1\mapsto\{a,b,c\},\,
      2\mapsto\{d\},\,
      3\mapsto\empti,\,
      4\mapsto\{a\},\,
      5\mapsto\{d\}
    \bigr\}
  \Bigr)
\end{align*}
of a $(1,4)$-HRA, we
map it to the configuration
\begin{align*}
  \Bigl(
    \bigl(
      q,
      \bigl\{
        1\mapsto\{2\},\,
        2\mapsto 1,\,
        3\mapsto\empti,\,
        4\mapsto 2,\,
        5\mapsto 1
      \bigr\}
    \bigr), \;
    \bigl\{
      1\mapsto\{b,c\}
    \bigr\}
  \Bigr)
\end{align*}
of a $(1,0)$-HRA which incorporates the non-nominal part of our representation scheme in its state.
Clearly, the state space of the new automaton in this simulation will experience an exponential blowup, as the next result shows. However, no additional histories will be needed, which is the main target here.

\begin{proposition}\label{prop:remove-regs-withstates}
Let $\AA=\langle Q,q_0,H_0,\delta,F\rangle$ be an $(m,n)$-HRA.
We can construct an $(m,0)$-HRA $\AA'=\langle Q',q'_0,H'_0,\delta',F'\rangle$
  that is empty if and only if $\AA$~is empty.
We have $|Q'| \in O(2^{mn} n B_n |Q|)$ and $|\delta'| \in O(2^{mn} n B_n |\delta|)$,
  where $B_n$~is the $n$th Bell number.
Moreover, $\AA'$ contains reset transitions if and only if $\AA$~contains reset transitions.
%The construction takes $O\bigl((m+n)(|Q'|+|\delta'|)\bigr)$ time.\nt{Same comment as before on time}
\end{proposition}
\begin{proof}
Each state $q \in Q$ will be simulated by several states $(q,f) \in Q'$,
  where $f : [m + n] \to \PP([n])$ will be called an assignment \emph{skeleton}.
Such a skeleton~$f$ 
is \emph{valid} when:
\begin{itemize}
\item $|f(i)| \le 1$ for registers $i \in [m+1,m+n]$;
\item $f(i)\subseteq\bigcup_{j=1}^n f(m+j)$ for histories $i \in [m]$;
\item for all $k \in [n]$ there is a $k' \in [k]$ such that $\bigcup_{i=1}^{k} f(m+i) = [k']$.
\end{itemize}
The latter condition essentially stipulates that $f$ is a partition function on the set $[m+1,m+n]$: the elements of the set are uniquely assigned numbers which can be seen as class indices\,---\,two elements are assigned the same number iff they belong to the same class. There is also a special class in this partition, namely of all elements of $[m+1,m+n]$ to which $f$ assigns $\empti$.

We can now define the rest of $\AA'$. First, we let
$q_0'=(q_0,f_0)$, where $(f_0,H_0')$ is the symbolic representation of $H_0$.
In order to construct $\delta'$
we define a transition relation on skeletons,
  which is very similar to the configuration graph of HRAs (\autoref{def:confgraph})
  except that it allows symbols to be permuted after the transition is taken.
We write $f \trdelta{X,X'} f'$ when
  there exists a permutation~$\pi$ on $[n]$ and a~$k\in[n]$ such that
  $k \in f \at X$ and $f'=\pi\circ(f[\move{k}{X'}])$.
We write $f \trdelta{X} f'$ when
  there exists a permutation~$\pi$ such that $f'=\pi\circ(f[X \mapsto \empti])$.

To simulate one transition of the form $q \trdelta{X,X'} q'$ from~$\delta$,
  we use several transitions of the form $(q,f) \trdelta{\ell} (q',f')$ in~$\delta'$.
Let us consider an arbitrary pair $(f,f')$ of valid skeletons, and see how to pick~$\ell$.
There are four cases, depending on whether $X$~and~$X'$ mention or not registers.
\begin{itemize}
\item Case $X \subseteq [m]$ and $X' \subseteq [m]$.
  It must be that $f \trdelta{\empti,\empti} f'$, and we pick $\ell=(X,X')$.
\item Case $X \subseteq [m]$ and $X' \not\subseteq [m]$.
  It must be that $f \trdelta{\empti,X'} f'$, and we pick $\ell=(X,\empti)$.
\item Case $X \not\subseteq [m]$ and $X' \subseteq [m]$.
  It must be that $f \trdelta{X,\empti} f'$, and we pick $\ell=(\empti,X')$.
\item Case $X \not\subseteq [m]$ and $X' \not\subseteq [m]$.
  It must be that $f \trdelta{X,X'} f'$, and we pick $\ell=(\empti,\empti)$.
\end{itemize}
Similarly, each reset transition $q \trdelta{X} q'$ from~$\delta$ yields
several transitions of the form $(q,f) \trdelta{Z} (q',f')$ in~$\delta'$.
Given an arbitrary pair $(f,f')$ of valid skeletons, we pick~$Z$ as follows.
If $X \subseteq [m]$ then it must be that $f'=f$ and we pick $Z=X$.
Otherwise, $f \trdelta{X} f'$ and we pick $Z=\empti$.

To estimate $|Q'|$ it suffices to count how many valid skeletons there are.
The values $f(m+1),\ldots,f(m+n)$ of a valid skeleton
  correspond to a partition of the registers
  and a selection of a class (if any) whose registers are empty.
There are $B_n$ possible partitions and $\le(n+1)$ possible selections,
  which gives $\le (n+1)B_n$ cases.
For the values $f(1),\ldots,f(m)$ of a valid skeleton
  there are $\le 2^{mn}$ possibilities.
In total, $|Q'|\le 2^{mn} (n+1) B_n |Q|$.

%The $[m+1,m+n]$ component of each skeleton $f$ corresponds to just such a partition and a selection of which class (if any) is special ($f$ assigns $\empti$ to all its elements). This gives $\le (n+1)B_n$ cases. Combined with the 
%$2^{mn}$ choices for $f\upharpoonright[m]$, we obtain $2^{mn}(n+1)B_n$.The $n$ registers can be partitioned in equivalence classes in $B_n$ ways,
%  by the definition of the Bell number.
%We can identify these classes by the smallest register they contain.
%One of the classes might be empty: $\le n+1$ cases.\nt{I think what we want to say here is that one of the classes may have all its registers empty. I would suggest: ``The $[m+1,m+n]$ component of each skeleton $f$ corresponds to just such a partition and a selection of which class (if any) is special ($f$ assigns $\empti$ to all its elements). This gives $\le (n+1)B_n$ cases. Combined with the 
%$2^{mn}$ choices for $f\upharpoonright[m]$, we obtain $2^{mn}(n+1)B_n$.''
%\rg{I wrote this paragraph in a hurry, but, yes, that's what I wanted to say.
%I'll rewrite it. Did we define $\upharpoonright$?}}
%For each history, we need to specify which classes they contain:
%  $\le 2^n$~cases for each of the $m$~histories.
%In total, $|Q'|\le B_n (n+1) 2^{mn}$.
%\nt{I think the old computation was correct as well (but not as tight): we have $(n+1)!$ choices for $f\upharpoonright[m+1,m+n]$: $f$ can assign to $m+1$ one of two values ($\empti,\{1\}$), to $m+2$ on of 3 values, and so on, and to $m+n$ one of $n+1$ values. Combined with the $2^{mn}$ choices for $f\upharpoonright[m]$, we obtain $2^{mn}(n+1)!$}

To estimate $|\delta'|$, note that once $f$~is fixed in the construction above,
  the constraints on $f'$ determine it uniquely.
So, the number of transitions increases by the same factor as the number of~states.
\end{proof}

Since $\log B_n \in \Theta(n \log n)$,
  we have that $\log\bigl(2^{mn} (n+1) B_n\bigr) \in \Theta(mn + n\log n)$.
%\nt{Since this is going in the exponent, maybe big-Theta is not good enough and we shoud use asymptotically-$<$ instead?}
% >>>

% vim:spell:spelllang=en_gb:fmr=<<<,>>>:

\section{Emptiness and Universality}\label{sec:empty}

\subsection{Emptiness}

Here we show that deciding emptiness is \textsc{Ackermann}-complete.
We work by reducing from and to state reachability problems in counter systems (similarly e.g.~ to~\cite{DA,CMA}).
For the upper bound,
  we reduce nonemptiness of HRAs to control-state reachability of T-VASSs.
For the lower bound,
  we reduce control-state reachability of R-VASSs to nonemptiness of HRAs.
Recall that 
the \boldemph{nonemptiness problem} for HRAs asks, given an HRA~$\AA$ with initial state $q_0$ and initial assignment~$H_0$, whether
  $(q_0,H_0) \trdeltaa{w} (q_F,H_F)$ for some word~$w$, final state $q_F$ and assignment~$H_F$.

The configurations of a $k$-dimensional TR-VASS
  (Transfer--Reset Vector Addition System with States)
  have the form $(q,\vec{v})$,
  where $q$~is a state from a finite set,
  and $\vec{v}$~is a $k$-dimensional vector of nonnegative counters.
A VASS has moves that shift the counter vector,
  changing $\vec{v}$ into $\vec{v}+\vec{v}\,{}'$,
  where $\vec{v}\,{}'$~comes from some finite and fixed subset of~$\Z^k$.
An R-VASS also has moves that reset a counter,
  changing $\vec{v}$ into $\vec{v}[i\mapsto0]$ for some counter~$i$.
A T-VASS also has moves that transfer the content of one counter into another counter,
  changing $\vec{v}$ into $\vec{v}[j\mapsto \vec{v}(i)+\vec{v}(j)][i\mapsto0]$
  for some~$i \ne j$.
A TR-VASS is the obvious combination of the above, and is formally defined as follows.

\begin{definition}[TR-VASS]
A $k$-dimensional \boldemph{Transfer-Reset Vector Addition System with States}~$\AA$
   is a pair $\langle Q,\delta \rangle$,
   where $Q$ is a finite set of states,
   and $\delta\subseteq Q\times (\Z^k\uplus[k]^2\uplus[k])\times Q$ is a transition relation.
A \boldemph{configuration} of $\AA$ is a pair $(q,\vec{v})$
  of a state~$q$ and a vector $\vec{v}\in\N^k$ of counter values.
The \boldemph{configuration graph} of $\AA$ is constructed by including an arc $(q,\vec{v})\to(q',\vec{v}\,{}')$ when one of the following holds: 
\begin{itemize}
\item there is some $(q,\vec{v}\,{}'',q')\in\delta$ such that $\vec{v}\,{}'=\vec{v}+\vec{v}\,{}''$
\item there is some $(q,(i,j),q')\in\delta$ such that
  $\vec{v}\,{}'=\vec{v}[i\mapsto0][j\mapsto \vec{v}(i)+\vec{v}(j)]$ and $i \ne j$
\item there is some $(q,(i,i),q')\in\delta$ and $\vec{v}\,{}'=\vec{v}$
\item there is some $(q,i,q')\in\delta$ such that $\vec{v}\,{}'=\vec{v}[i\mapsto 0]$.
\end{itemize}
The \boldemph{control-state reachability problem} for $\AA$ asks whether, given states $q_0,q_F$ and initial vector $\vec{v}_0$, is there some $\vec{v}_F$ such that
  $(q_0,\vec{v}_0) \trdeltaa{} (q_F,\vec{v}_F)$.
\end{definition}

The reduction from a $(m,0)$-HRA to a T-VASS of dimension $2^m-1$ is done by mapping each nonempty set of histories $X$ into a counter $\widetilde{X}$ of the corresponding T-VASS. Then, name-accepting transitions are mapped into counter decreases and increases, while resets result in transfers between the counters.
Let $\widetilde{\,{\cdot}\,}:\PP\bigl([m]\bigr)\to[0,2^m-1]$ be a bijection
  such that $\widetilde{\emptyset}=0$;
  for instance, one could take $\widetilde{X}\defeq\sum_{i\in X} 2^{i-1}$.
Further, given an assignment~$H$, let $\widetilde{H}$ denote the vector
  $(h_1,\ldots,h_{2^m-1})\in\N^{2^m-1}$ such that $h_{\widetilde{X}}=|H@X|$,
  for all nonempty $X\subseteq[m]$;
that is, $h_{\widetilde{X}}$~counts how many names occur
  in exactly the histories indexed by~$X$.

\begin{lemma}\label{lem:hra-to-trvass}
Given an $(m,0)$-HRA~$\AA$ it is possible
  to construct a T-VASS~$\AA'$ of dimension $2^m-1$ such that,
  for all $q,q',H,H'$,
\[\text{
  $\exists w,\, (q,H) \trdeltaa{w}_{\AA} (q',H')$
  \qquad if and only if \qquad
  $(q,\widetilde{H}) \too_{\AA'} (q',\widetilde{H'})$.
}\]
Let $Q$~and~$\delta$ be the states and the transitions of~$\AA$,
  and let $Q'$~and~$\delta'$ be the transitions of~$\AA'$.
We have that $|Q'| \in O(2^m |Q|)$ and $|\delta'| \in O(2^m |\delta|)$.
Moreover, the construction takes $O(|Q'|+m|\delta'|)$ time.
If there are no reset transitions in~$\AA$,
  then $\AA'$ is a $|\delta|$-dimensional VASS with $Q'=Q$ and $|\delta'|=|\delta|$
  that uses only increments and decrements.
\end{lemma}
Let $\vec{0}$ be the all-zero vector $(0,\ldots,0)$.
Let $\vec{\delta}_i$ be $\vec{0}[i\mapsto1]$ for $i\in[m]$,
and $\vec{\delta}_0$ be $\vec{0}$.
\begin{proof}
For each transition $q \trdelta{X,X'} q'$ of the HRA, we construct a transition
  $q\xtrdelta{\vec{\delta}_{\widetilde{X}'}-\vec{\delta}_{\widetilde{X}}}q'$
  in the T-VASS\null.
For each transition $q \trdelta{X} q'$ of the HRA, we construct a path
\begin{align*}
  q
  \trdelta{1,j_1} \cdot
  \trdelta{2,j_2} \cdot
  \trdelta{3,j_3} \cdots
  \xtrdelta{2^m-2,j_{2^m-2}} \cdot
  \xtrdelta{2^m-1,j_{2^m-1}} q'
\end{align*}
  in the T-VASS such that $j_{\widetilde{Y}}=\widetilde{Y \setminus X}$.
To construct such a path we iterate through $2^m-1$ nonempty sets~$Y$,
  and for each we compute $Y \setminus X$ in $O(m)$~time.
\end{proof}

\autoref{lem:hra-to-trvass} implies that
nonemptiness of a HRA reduces to control-state reachability of a T-VASS.
We shall describe an algorithm that solves control-state reachability for the T-VASS
  constructed in \autoref{lem:hra-to-trvass}.
The analysis of this algorithm depends on the so-called Length Function Theorem,
  which is phrased in terms of the Fast Growing Hierarchy and bad sequences.
We define these next.

The Fast Growing Hierarchy
  consists of classes $\FF_0$, $\FF_1$, $\FF_2$, \dots\ of functions,
    where $\FF_0=\FF_1$ contain the linear functions,
    $\FF_2$ contains the elementary functions,
    primitive recursive functions are in $\FF_k$ for some finite~$k$,
    and $\FF_{\omega}$ is the \textsc{Ackermann} complexity class.
The classes $\FF_k$ are defined in terms of the following functions:
\begin{align*}
  F_0(x) &\defeq x + 1
&
  F_{n+1}(x) &\defeq (\underbrace{F_n\circ\cdots\circ F_n}_{\text{$x+1$ times}})(x)
    = F_n^{x+1}(x)
&
  F_{\omega}(x) &\defeq F_x(x)
\end{align*}
For $k \ge 2$,
  (a)~$f \in \FF_k$ if and only if $f \in O(F_k^n)$ for some~$n$; and
  (b)~a nondeterministic algorithm
      using space bounded by some function in~$\FF_k$
  can be transformed into a deterministic algorithm
    using time bounded by some (other) function in~$\FF_k$.

Let $X$ be a partially ordered set with some size function $|{\cdot}|:X\to\mathbb{N}$. 
We say that a sequence $x_0, x_1, x_2,\ldots$ of elements of $X$ is a \boldemph{bad sequence} when $x_i \not\le x_j$ for all $i < j$.
Given a strictly increasing function $g:\mathbb{N}\to\mathbb{N}$, we say that the sequence is \boldemph{controlled} by~$g$ when $|x_{i+1}| \le g(|x_i|)$ for all $i$.
We will consider such sequences of VASS configurations, where the order is given by
\begin{align*}
(q,\vec{v}) \le (q',\vec{v}')
  \qquad\text{iff}\qquad
q=q'
  \;\land\; \vec{v}(1) \le \vec{v}'(1)
  \;\land\; \vec{v}(2) \le \vec{v}'(2)
  \;\land\; \vec{v}(3) \le \vec{v}'(3)
  \;\land\; \ldots
\end{align*}

\begin{lemma}[Length Function Theorem \cite{algo-wqo}]\label{th:length-function}
Let $\hat{q}_0, \hat{q}_1, \hat{q}_2, \ldots$ be a bad sequence of $k$-dimensional VASS configurations.
If the sequence is controlled by some function $g\in\FF_{\gamma}$ with $\gamma \ge 1$,
  then its length is bounded by $f(|\hat{q}_0|)$ for some function $f\in\FF_{\gamma+k}$.
\end{lemma}

We can now describe and analyze an algorithm for deciding emptiness of a $(m,0)$-HRA\null.

\begin{proposition}\label{prop:emptiness-ub}
The emptiness problem for $(m,0)$-HRAs is in $\FF_{2^m}$ when $m>0$.
Thus, the emptiness problem is in~$\FF_{\omega}$ when $m$~is part of the input.
\end{proposition}
\begin{proof}
Let $\AA$ be the given HRA,
  and let $\AA'$ be the T-VASS constructed as in \autoref{lem:hra-to-trvass}.
We use the backward coverability algorithm~\cite[Sections 1.2.2 and 2.2.2]{algo-wqo},
  which explores all bad sequences
$
  (q_0,\vec{v}_0),\ (q_1,\vec{v}_1),\ \ldots,\ (q_L,\vec{v}_L)
$
such that
\begin{itemize}
\item $(q_0,\vec{v}_0)$ is a minimal final configuration, and
\item $(q_k,\vec{v}_k)$ is a minimal configuration out of those
  that can reach a configuration~$\ge (q_{k-1},\vec{v}_{k-1})$.
\end{itemize}
%\rg{In \cite[Section 1.2.2]{algo-wqo} one can find a correctness proof for the algo
%  when applied to WSTS\null.
%It is trivial to show that TR-VASS is WSTS, but I think this isn't explicit in~\cite{algo-wqo}.
%Maybe we should include it.}
The constraint that $(q_0,\vec{v}_0)$ is a minimal final configuration simply means
  that $q_0$~is final and $\vec{v}_0=\vec{0}$.
To construct such sequences,
  we need an effective way of generating all possible $(q_k,\vec{v}_k)$,
  given a fixed $(q_{k-1},\vec{v}_{k-1})$.
For this, we enumerate all transitions of $\AA'$ that go to~$q_{k-1}$.
There are two types of such transitions:
  $q_k \trdelta{\vec{\delta}_i-\vec{\delta}_j}q_{k-1}$ and $q_k\trdelta{i,j}q_{k-1}$.
For $q_k\trdelta{\vec{\delta}_i-\vec{\delta}_j}q_{k-1}$,
  we let $\vec{v}_k$ be $\max(\vec{v}_{k-1}-\vec{\delta}_i+\vec{\delta}_j,\vec{0})$,
  where $\max$~is taken pointwise.
For $q_k \trdelta{i,j} q_{k-1}$ with $i=j$, we take $\vec{v}_k$ to equal $\vec{v}_{k-1}$.
For $q_k \trdelta{i,j} q_{k-1}$ with $i \ne j$, we may have multiple choices for~$\vec{v}_k$.
Assuming $\vec{v}_{k-1}(i)=0$,
  it could be that $\vec{v}_k(i)$ is any of $0,1,\ldots,\vec{v}_{k-1}(j)$;
otherwise, if $\vec{v}_{k-1}(i) \ne 0$, the transition could not have been taken.
In all the cases from above,
  we keep only those choices of~$\vec{v}_k$ that ensure the sequence is bad.

To show that the sequences so constructed are finite, we use \autoref{th:length-function}.
Let $|(q,\vec{v})|$ be the number bits in a concrete representation of $(q,\vec{v})$:
  we encode $q$ with $\sim\log_2 |Q'|=m\log_2|Q|$ bits,
  then we write each $\vec{v}(i)$ in binary,
    precede each of its bits by~$1$ and mark the end with~$0$.
(For example, we represent $5$ by $1110110$.)
For the sequences constructed as in the previous paragraph,
  we have $|(q_k,\vec{v}_k)| \le 2 \cdot |(q_{k-1},\vec{v}_{k-1})|$.
Thus, the sequences are controlled by $g(x)=2x$, which is a function in~$\FF_1$.
As $\AA'$ has dimension $2^m-1$,
  \autoref{th:length-function} gives us
  that the length~$L$ of the sequence is bounded by some function in~$\FF_{2^m}$.

A nondeterministic algorithm can repeatedly guess the correct successor in the sequence,
  using $2^L \cdot |(q_0,\vec{0})|$ space.
If $m\ge 1$, then this is bounded by $f(m\log|Q|)$ for some function $f \in \FF_{2^m}$,
  and we are in a situation where
    the distinctions time\slash space and deterministic\slash nondeterministic
    are~irrelevant.
\end{proof}

It is possible to modify the algorithm described in the previous proof
  so that it works directly on the HRA representation,
  without appealing to \autoref{lem:hra-to-trvass}.
Similarly, it is possible to extend the algorithm described in the previous proof
  to handle registers directly, without appealing to \autoref{prop:regs_his}.
Such improvements may be worthwhile in an implementation,
  but the complexity upper bound remains Ackermannian.

\medskip

Doing the opposite reduction we show that deciding emptiness is \textsc{Ackermann}-hard
  even for strongly deterministic HRAs.
In this direction, each R-VASS of dimension~$m$ can be simulated by an $(m,0)$-HRA so that the value of each counter $i$ of the former is the same as the number of names appearing precisely in history $i$ of the latter.
In order to extend the bound to strongly deterministic HRAs one can choose to reduce from a restricted class of R-VASSs, so that the image of the reduction can be made strongly deterministic, or resolve nondeterminacy at the level of HRAs by appropriate obfuscation. We follow the latter, simpler solution.
%The only slight obstacle is in obtaining an HRA that is strongly deterministic, but the $n$ registers are used precisely to that effect, i.e.\ for resolving non-determinism.

\begin{proposition}\label{prop:R-VASS-to-HRA}
The emptiness problem for strongly deterministic HRAs is \textsc{Ackermann}-hard.
\end{proposition}
\begin{proof}
Let $\AA$ be an $m$-dimensional R-VASS
  whose additive transitions only increment or decrement single counters:
  for each transition $q \trans{\vec{v}} q'$,
  we have $\vec{v} = \pm \vec{\delta}_i$ for some~$i$.
By~\cite{Schnoebelen:2010}, control-state reachability for such R-VASSs is \textsc{Ackermann}-hard.
%\footnote{The argument presented  in~\cite{Schnoebelen:2010} assumes Minsky machines which only branch in the form ``$\texttt{if c=0 then goto l; c--;}$" and are therefore deterministic in the above sense, while Schnoebelen's auxiliary constructions are also deterministic.}
We construct an $(m,0)$-HRA $\AA'$ with the same states as $\AA$, and
we map:
  each $q\xtrdelta{\delta_i}q'$ to $q\trdelta{\empti,\{i\}}q'$,
  each $q\xtrdelta{-\delta_i}q'$ to $q\trdelta{\{i\},\empti}q'$,
  and each $q\xtrdelta{i}q'$ to $q\trdelta{\{i\}}q'$.
We can see that $\AA'$ simulates the behaviour of $\AA$
  by storing the value of each counter~$i$ as $|H \at \{i\}|$.
Hence,
  $\LL(\AA')$~is nonempty if and only if
  $q_F$ is reachable, from $(q_0,\vec{v}_0)$, in the R-VASS~$\AA$.

We observe that $\AA'$ may not be strongly deterministic. Suppose that the size of the transition function of $\AA$ is $n$. We can then impose strong determinacy on $\AA'$ by enriching it with $n$ registers and preluding each transition of the above translation with a transition reading from one of the additional registers. We thus obtain an $(m,n)$-HRA that is strongly deterministic and simulates $\AA$ as above.
\end{proof}

\begin{proposition}\label{prop:emptiness-general}
The emptiness problem of HRAs is \textsc{Ackermann}-complete.
\end{proposition}
\begin{proof}
Combine \autoref{prop:regs_his}
  with \autoref{prop:emptiness-ub} and \autoref{prop:R-VASS-to-HRA}.
\end{proof}

\subsection{Universality}

We finally consider universality and language containment.
Note first that our machines inherit undecidability of these properties from register automata~\cite{RA2}.
However, these properties are decidable in the deterministic case.

In order to simplify our analysis, we shall be reducing HRAs to the following compact form where  $\epsilon$-transitions are incorporated inside name-accepting ones. 
As we show below, no expressiveness is lost by this packed form.

A \emph{packed $(m,0)$-HRA} is a tuple $\AA=\langle Q,q_0,\delta,H_0,F\rangle$ defined exactly as an $(m,0)$-HRA, with the exception that now:
\[
\delta \subseteq Q\times \PP([m])\times \PP([m])\times \PP([m])\times Q
\]
We shall write $q\xrightarrow{Y;X,X'}q'$ for $(q,Y,X,X',q')\in\delta$. The semantics of such a transition is the same as that of a pair of transitions $q\trdelta{Y}\cdot\trdelta{X,X'}q'$ of an ordinary HRA\null. Formally, configurations of packed HRAs are pairs $(q,H)$, like in HRAs, and the configuration graph of a packed HRA $\AA$ like the above is constructed as follows. We set $(q,H)\trdelta{a}(q,H')$ if there is some $q\xrightarrow{Y;X,X'}q'$ in $\delta$ such that, setting $H_Y=H[Y\mapsto\empti]$, we have
$a\in H_Y\at X$ and $H'=H_Y[\move{a}{X'}]$.

\begin{lemma}
Let $\AA$ be an $(m,0)$-HRA. There is a packed $(m,0)$-HRA $\AA'$ such that $\AA\sim\AA'$.
\end{lemma}
\begin{proof}
Let $\AA=\langle Q,q_0,\delta,H_0,F\rangle$. We set $\AA'=\langle Q,q_0,\delta',H_0,F'\rangle$ where:
\begin{align*}
F' &= \{ q'\in Q\ |\ \exists q\in F,Y.\ q'\trdeltaa{Y}q\in\delta\}\\
\delta' &= \{(q,Y,X,X',q')\ |\ q\trdeltaa{Y}\cdot\trdelta{X,X'}q'\in\delta\}
\end{align*}
Bisimilarity of $\AA$ and $\AA'$ is witnessed by the identity on
configurations, which means that $R = \{\,((q,H),(q,H))\mid q\in Q \land H\in\his\,\}$ is a bisimulation.
\end{proof}

We shall decide language containment via complementation. In particular, given a deterministic packed HRA $\AA$, the automaton $\AA'$ accepting the language $\names^*\setminus\LL(\AA)$ can be constructed in the analogous way as for deterministic finite-state automata, namely by obfuscating the automaton with all missing transitions and swapping final with non-final states. %Finding the missing transitions is easy: for each state $q$ and each set $X$ such that there is no transition of the form $q\xtrdelta{Y;X\setminus Y,X'}q'$ in $\AA$, we add a transition $q\xtrdelta{\empti;X,\empti}q_F$ to some sink final state $q_F$.

\begin{lemma}\label{lem:complement}
Deterministic packed HRAs are closed under complementation.
\end{lemma}
\begin{proof}
Let $\AA=\langle Q,q_{0},\delta,H_{0},F\rangle$ be a packed $(m,0)$-HRA.
Following the above rationale, we construct a packed $(m,0)$-HRA $\AA'=\langle Q\uplus\{q_F\},q_0,\delta\cup\delta',H_0,F'\rangle$, 
where $F'=\{q_F\}\cup(Q\setminus F)$ and $\delta'$ is given as follows.
For each $q\in Q$ and all $X$ such that there is no $q\xtrdelta{Y;X\setminus Y,X'}q'$ add a transition $q\xtrdelta{\empti;X,\empti}q_F$ in $\delta'$. In addition, $\delta'$ contains a transition $q_F\xtrdelta{[m];\empti,\empti}q_F$.
\\
We claim that $\LL(\AA')=\names^*\setminus\LL(\AA)$. Indeed, if $s\in\LL(\AA')$ and $s$ is accepted at a state in $Q\setminus F$ then, since $\AA$ is deterministic, we have $s\notin\LL(\AA)$. Otherwise, if $s=s'as''$ with $a$ the point where a transition to the sink state is taken then, upon acceptance of $s'$ by $\AA$, $a$ appears precisely in some histories $X$ such that $\AA$ has no transition to accept $a$ at that point. Thus, $s\notin\LL(\AA)$. 
\\
Conversely, if $s\in\names^*\setminus\LL(\AA)$ then 
either $s$ induces a configuration in $\AA$ which does not end in a final state, or 
$s=s'as''$ where $s'$ is accepted by $\AA$ but at that point $a$ is not a possible transition. We can see that, in each case, $s\in\LL(\AA')$.
\end{proof}

\begin{proposition}\label{prop:lang-inclusion}
Language containment and universality are undecidable for (general) HRAs
  and \textsc{Ackermann}-complete for strongly deterministic HRAs.
\end{proposition}

\begin{proof}
Undecidability in the general case is inherited from RAs.

Now consider two HRAs $\AA$~and~$\AA'$ such that we can compute the complement of~$\AA'$.
Then,
  we can decide the language containment $\LL(\AA)\subseteq\LL(\AA')$
  by checking whether the product of $\AA$ with the complement of~$\AA'$ is empty.
The product construction is polynomial,
  and the emptiness check is in \textsc{Ackermann} (\autoref{prop:emptiness-ub}).
Thus, language containment is in \textsc{Ackermann}
  if computing the complement of~$\AA'$ is in \textsc{Ackermann}.
This is the case because
  (a)~removing registers can be done while preserving determinism
    with only an exponential increase in size (\autoref{prop:regs_his}), and
  (b)~complementing deterministic HRAs without registers takes polynomial time
    (\autoref{lem:complement}).
%\rg{TODO: check these points carefully.}
For hardness,
  note that emptiness and universality are equally hard in the deterministic case
    (\autoref{lem:complement}),
  and emptiness is \textsc{Ackermann}-hard (\autoref{prop:R-VASS-to-HRA}).

We showed that
  language containment is in \textsc{Ackermann}
  and universality is \textsc{Ackermann}-hard.
Finally, note that there is a trivial reduction from universality to language containment.
\end{proof}

%%% Local Variables: 
%%% mode: latex
%%% TeX-master: "jour"
%%% End: 
% vim:spell:spelllang=en_us:

\section{Weakening HRAs}\label{sec:weak}
% prelude <<<

Since the complexity of HRAs is substantially high, e.g.~for deciding emptiness, it is useful to seek for restrictions thereof which allow us to trade expressiveness for efficiency.
As the encountered complexity stems from the fact that HRAs can simulate computations of R-VASSs,
our strategy for producing weakenings is to restrict the functionalities of the corresponding R-VASSs. We follow two directions:

\begin{enumerate}[label=\({\alph*}]
\item We remove reset transitions. This corresponds to removing counter transfers and resets and drops the complexity of control-state reachability to exponential space.
\item We restrict the number of histories to just one. We thus obtain polynomial space complexity as the corresponding counter machines are simply one-counter automata. This kind of restriction is also a natural extension of FRAs with history resets.
\end{enumerate}
Observe that each of the aspects of HRAs targeted above corresponds to features (1,2) we identified in the Introduction, witnessed by the languages $\LL_1$ and $\LL_2$ respectively. We shall see that each restriction leads to losing the corresponding language.

%Our analysis on emptiness for general HRAs from \autoref{sec:empty} is not applicable to these weaker machines as we now need to take registers into account: the simulation of registers by histories is either not possible or not practical for deriving satisfactory complexity bounds.
%Additionally,
%a direct analysis will allow us to reduce instances of counter machine problems to our setting decreasing the complexity size by an exponential, compared to our previous reduction.
%
%Solving emptiness for each of the weaker versions of HRAs will involve reduction to a name-free counter machine. In both cases, the reduction shall follow the same concept as in \autoref{sec:empty}, namely of simulating computations with names \emph{symbolically}.

% >>>
\subsection{Non-reset HRAs} % <<<

We first weaken our automata by disallowing resets. We show that the new machines retain all their closure properties apart from Kleene-star closure. The latter is concretely manifested in the fact that language $\LL_1$ of the Introduction is lost. On the other hand, the emptiness problem reduces in complexity to exponential space.

\begin{definition}
A \boldemph{non-reset HRA} of type $(m,n)$ is an $(m,n)$-HRA $\AA=\langle Q,q_0,H_0,\delta,F\rangle$ such that there is no $q\trdelta{X}q'\in\delta$.
\end{definition}

\subsubsection*{Closure properties}
Of the closure constructions of \autoref{sec:closure} we can see that union and intersection readily apply to non-reset HRAs, while the construction for concatenation needs some amendments.

More specifically, of the two constructions presented in the proof of \autoref{prop:closure}, the one for concatenation can be adapted to non-reset HRAs as follows.
We add empty transitions from the final states of $\AA_1'$ to the initial state of a version of $\AA_2'$ which keeps the places used by $\AA_1'$ untouched and uses its own separate copy of places, obfuscating its own transitions so as to capture accidental matchings of the legacy names of $\AA_1'$.
This solution cannot be used for Kleene closure as in each loop the automaton needs to find a fresh copy of its initial configuration, and be able to use it (in the previous construction, the final assignment of $\AA_1'$ is lost).

On the other hand, using an argument similar to that of~\cite[Proposition~7.2]{CMA}, we can show that the language $\LL_1$ is not recognised by non-reset HRAs and, hence, the latter are not closed under Kleene star. Finally, note that the HRA constructed for the language $\LL_4$ in \autoref{ex:complement} is a non-reset HRA, which implies that non-reset HRAs are not closed under complementation.

\subsubsection*{Emptiness}

In the general case we saw an upper bound of $\FF_{2^m}$ (\autoref{prop:emptiness-ub}),
  by a reduction to T-VASS followed by the backward coverability algorithm.
For a non-reset HRA, the same reduction yields a VASS, without transfers.
In the absence of transfers,
  better bounds are known for the backward coverability algorithm%
  ~\cite{backward-cover-analysis}.
More generally,
  it has been known for some time that coverability for VASS
  is \textsc{ExpSpace}-complete \cite{Rackoff,Lipton}.

The following result refers to the number~$N$ of bits used to represent an HRA\null.
Of course, $N$~depends on the exact representation being used.
Still, we do not make this representation explicit
  because the result holds for a wide variety of possible representations.
We only require that the representation obeys $m,n,|\delta|,\log|Q|\in O(N)$.

\begin{proposition}\label{prop:nrHRA-ub}
The emptiness problem for a non-reset HRA is in $\textsc{ExpSpace}$.
More precisely, it is in $\textsc{NSpace}\bigl(2^{O(N \log N)}\bigr)$,
  where $N$~is the number of bits used to represent the HRA\null.
\end{proposition}

\begin{proof}
We start with an $(m,n)$-non-reset-HRA $\AA=\langle Q,q_0,\delta,H_0,F\rangle$.
We use \autoref{prop:remove-regs-withstates} to construct an $(m,0)$-non-reset-HRA
  $\AA'=\langle Q',q'_0,\delta',H'_0,F' \rangle$ that preserves emptiness.
Moreover, $\log|Q'|\in O(mn+n\log n+\log|Q|)$.
Using \autoref{lem:hra-to-trvass},
  we reduce $\AA'$ to a $(|\delta|+1)$-dimensional VASS with $|Q'|$~states
  that uses only increments\slash decrements.
(\autoref{lem:hra-to-trvass} creates an $m$-dimensional VASS
  where $m$ is the number of sets labelling transitions.
  \autoref{prop:remove-regs-withstates}
    puts in $\AA'$ only sets that already occurred in $\AA$,
    with the possible exception of~$\empti$.)
Now we apply the backward coverability algorithm,
  as described in the proof of \autoref{prop:emptiness-ub}.
By \cite[Theorem~2]{backward-cover-analysis},
  the algorithm will only consider counter values less than some
  $V=(3|Q'|)^{2^{O(|\delta|\log|\delta|)}}$.
In the nondeterministic version of the algorithm,
  we guess the next configuration,
  which means we only need space $O((|\delta|+1)\log V)$ to store a couple of configurations.
We have
\begin{align*}
  \log \log V
    = O(|\delta|\log|\delta| + \log\log|Q'|)
    = O(|\delta|\log|\delta| + \log(mn+n\log n + \log|Q|))
\end{align*}
Since $m,n,|\delta|,\log|Q|\in O(N)$,
  we conclude $\log|\delta|+\log\log V \in O(N \log N)$.
This implies that a nondeterministic version of the backward coverability algorithm
  works in $\textsc{NSpace}\bigl(2^{O(N \log N)}\bigr)$.
\end{proof}

The previous proposition has a couple of obvious consequences.
First, emptiness is also in $\textsc{DSpace}\bigl(2^{O(N \log N)}\bigr)$,
  by Savitch's theorem.
Second, emptiness is also in $\textsc{DTime}\bigl(2^{2^{O(N \log N)}}\bigr)$,
  by a standard easy argument~\cite[Theorem~5.3]{oded-complexity}.
In fact, one can show that the time bound applies to the backward coverability algorithm,
  without invoking generic constructions from complexity theory:
By \cite[Theorem~2]{backward-cover-analysis},
  the runtime of the backward coverability algorithm%
  \,---\,like the counter values\,---\,%
  is also upper bounded by some $T=O\bigl((3|Q'|)^{2^{O(|\delta|\log|\delta|)}}\bigr)$.
The rest of the argument is as in the proof of \autoref{prop:nrHRA-ub}.

\begin{proposition}\label{prop:nrHRA-lb}
The emptiness problem for non-reset HRAs is \textsc{ExpSpace}-hard.
\end{proposition}

\begin{proof}
By~\cite{Lipton},
  the control-state reachability problem for VASS is $\textsc{ExpSpace}$-hard
  even if all the transitions are restricted to have labels of the form~$\pm\vec{\delta}_i$.
(More precisely,
  Lipton proves that certain parallel programs of size ${\it poly}(k)$
  can simulate any Turing machine that uses $<2^k$~space.
Then, \cite[Lemma~2]{Lipton} asserts that reachability in these programs
  reduces to reachability in VAS, the full proof being:
`We omit a detailed proof of this lemma.
It should, however, be clear that parallel programs can be encoded as vector addition systems.'
Similarly, we claim without proof, that it should be clear how Lipton's programs
  reduce to the control-state reachability problem for VASSs whose transitions only
  increment\slash decrement single counters.)
We shall reduce the control-state reachability problem for such VASSs
  to the emptiness problem for non-reset HRAs.

Let $m$~be the dimension of the VASS\null.
We construct a HRA with $m'$~histories,
  where $m'$ is the smallest integer such that $m \le 2^{m'}-1$.
As a result, there exists an injection $\phi : [m] \to \PPnz([m'])$;
  we fix arbitrarily one such injection.
\begin{itemize}
\item For each transition $q \trdelta{+\vec{\delta}_i} q'$ in the VASS,
  we include a transition $q \trdelta{\empti, \phi(i)} q'$ in the HRA\null.
\item For each transition $q \trdelta{-\vec{\delta}_i} q'$ in the VASS,
  we include a transition $q \trdelta{\phi(i), \empti} q'$ in the HRA\null.
\end{itemize}
This construction maintains the invariant $|H \at \phi(i)|=\vec{v}(i)$.
To establish the invariant, we pick the initial history assignment $H_0$ accordingly.
Finally, we set as final the state in whose reachability we are interested.

The reduction described above is clearly polynomial,
  from which it follows that emptiness of non-reset HRAs (even without registers)
  is $\textsc{ExpSpace}$-hard.
\end{proof}

\begin{proposition}\label{prop:nrHRA}
The emptiness problem for non-reset HRAs is \textsc{ExpSpace}-complete.
\end{proposition}
\begin{proof}
Immediate from \autoref{prop:nrHRA-ub} and \autoref{prop:nrHRA-lb}.
\end{proof}

% >>>
\subsection{Unary HRAs} \label{sec:unary} % <<<

Our second restriction concerns allowing resets but bounding the number of histories to just one. Thus, these automata are closer to the spirit of FRAs and, in fact, extend them by rounding up their history capabilities. We show that these automata require polynomial space complexity for emptiness and retain all their closure properties apart from intersection. The latter is witnessed by failing to recognise $\LL_2$ from the Introduction. We can see that extending this example to multiple interleavings we can show that intersection is in general incompatible with bounding the number of histories.

\begin{definition}
A $(1,n)$-HRA is called \boldemph{unary HRA} of $n$ registers.
\end{definition}

\cutout{The reduction of FRAs to unary HRAs is rather straightforward and we give only a sketch of it. In particular, for each FRA $\AA$ we can construct a bisimilar unary HRA $\AA'$ by the following mapping of transitions.\ntside{Define FRAs}
\begin{itemize}
\item $q\trdelta{i}q'$ gets mapped to $q\xtrdelta{\{1,i{+}1\},\{1,i{+}1\}}q'$;
\item $q\trdelta{i^\circledast}q'$ gets mapped to $q\xtrdelta{\empti,\{1,i{+}1\}}q'$;
\item $q\trdelta{i^\bullet}q'$ gets mapped to $q\xtrdelta{\{1\},\{1,i{+}1\}}q'$ and $q\xtrdelta{\empti,\{1,i{+}1\}}q'$.
\end{itemize}
The set of states remains the same, and the initial history assignment $H_0$ of $\AA'$ contains in positions $2,\ldots,n{+}1$ a copy of the initial register assignment $\sigma_0$ of $\AA$, while $H_0(1)=\{\sigma_0(i)\ |\ i\in\{2,\cdots,n{+}1\}\}$.

\begin{lemma}
For $\AA,\AA'$ as above, $\AA\sim\AA'$.
\end{lemma}}

In other words, unary HRAs are extensions of FRAs where names can be selectively inserted or removed from the history and, additionally, the history can be reset. These capabilities give us in fact a strict extension.

\begin{example}\label{ex:qwe}
The automata used in \autoref{ex:HRA} for $\LL_1$ and $\LL_3$ were unary HRAs. Note that neither of those languages is FRA-recognisable.
On the other hand, in order to recognise
$\LL_2$, an HRA would need to use at least two histories: one history for the odd positions of the input and another for the even ones. We can formalise an argument to show that $\LL_2$ is not recognisable by unary HRAs as follows.
\end{example}

\begin{proof}
Suppose $\LL_2=\LL(\AA)$ for some unary HRA $\AA$ of $n$ registers and let
\[
w= a_1b_1\ldots a_kb_kb_1a_1\cdots b_k a_k
\]
for $k=n+1$ and some pairwise distinct names $a_1,b_1,\ldots,a_k,b_k$. As $w\in\LL_2$, there is a path, say $p$, in $\AA$ which accepts $w$. We divide $p$ as $p_1p_2$ with $p_2$ accepting the second half of $w$. Let $\hat{p}=\hat{p}_1\hat{p}_2$ be the corresponding configuration path and let $(q',H')$ be the first configuration in $\hat{p}_2$. We set $S=\{a_1,b_1,\cdots,a_k,b_k\}\setminus\{a\ |\ a\in H'(i)\land i>1\}$ and do a case analysis on the labels of the form $(X,X')$ which appear in $p_2$ and accept names from $S$. Since names in $S$ do not appear in any $H'(i)$, for $i>0$, it must be that each such $X$ is either $\{1\}$ or $\empti$. We have the following cases.
\begin{itemize}
\item There are two such labels, say $(\{1\},X_i)$ and $(\{1\},X_j)$, accepting names $a_i$ and $b_j$ respectively. But this would imply that $\AA$ also accepts $w'$, where $w'$ is $w$ with these occurrences of $a_i$ and $b_j$ swapped, contradicting $\LL(\AA)=\LL_2$ (as $w'\notin\LL_2$).
\item There are two such labels, say $(\empti,X_i)$ and $(\empti,X_j)$, accepting names $a_i$ and $b_j$ respectively. In order for $\AA$ not to accept $w'$ ($w'$ as above), it is necessary that a reset transition with label $Y\ni1$ occurs between the two transitions. Suppose $i<j$. Then, since $k>n$, there is a name $a_{i'}$ which does not appear in any place after clearing $Y$. Thus, $(\empti,X_j)$ can accept $a_{i'}$ and complete the path $p$ by accepting a word $w'\notin\LL_2$. Dually if $j\leq i$.
\item Each $a_i\in S$ is accepted by a label $(\{1\},X')$, and each $b_j\in S$ by a label $(\empti,X')$. Let $a_i\in S$ be the last such accepted in $p_2$. This means that the rest of the path has length at most $2n$. Therefore, since $k>n$, there is a $b_j\in S$ accepted in $p_2$ before $a_i$. Let $(q,H)$ be the configuration just before accepting $b_j$.
In order for $\AA$ not to accept any $a_{i'}$ at that point, it must be that all $a_{i'}\in S$ appear in $H$. Since $|S|>n+1$, there exists $a_{i'}\in H(1)\cap S$ such that $a_{i'}\not=a_i$. But then, the transition accepting $a_i$ can accept $a_{i'}$ instead and lead to acceptance of a word $w'\not\in\LL_2$.
\end{itemize}
We therefore reach a contradiction in every case.
\end{proof}

\subsubsection*{Closure properties}
The closure constructions of \autoref{sec:closure} readily apply to unary HRAs, with one exception: intersection. For the latter, we can observe that $\LL_2=\LL(\AA_1)\cap\LL(\AA_2)$,
where
\begin{wrapfigure}{r}{.25\linewidth}\vspace{-3.5mm}
\parbox{.3\linewidth}{%
\begin{tikzpicture}[automaton]
\node[state,initial,accepting] (q0) {$q_0$};
\node[state] (q1) [right=of q0] {$q_1$};
\path[transition]
  (q0) edge node[above]{$\scriptstyle\empti\ta1$} (q1)
  (q1) edge[bend left] node[below]{$\scriptstyle\empti,\empti\,/\, 1,1$} (q0);
\end{tikzpicture}\qquad
\begin{tikzpicture}[automaton]
\node[state,initial,accepting] (q0) {$q_0$};
\node[state] (q1) [right=of q0] {$q_1$};
\path[transition]
  (q0) edge node[above]{$\scriptstyle\empti\ta\empti\,/\, 1,1$} (q1)
  (q1) edge[bend left] node[below]{$\scriptstyle\empti,1$} (q0);
%  (q0) edge[loop above] node{$\scriptstyle 2,1$} ()
%  (q1) edge[loop above] node{$\scriptstyle 1,2$} ();
\end{tikzpicture}}\vspace{-7.5mm}
\end{wrapfigure}
%We can see that their corresponding languages are:
%
$\LL(\AA_1) = \{ a_1a_1'\ldots a_na_n'\in\names^*\ |\ a_1\ldots a_n\in\LL_0\}$
and
$\LL(\AA_2) = \{ a_1a_1'\ldots a_na_n'\in\names^*\ |\ a_1'\ldots a_n'\in\LL_0\}$,
and
$\AA_1$ and $\AA_2$ are the unary $(1,0)$-HRAs on the side, with empty initial assignments.
On the other hand, unary HRAs are not closed under complementation as well, as one can construct unary HRAs accepting $\overline{\LL(\AA_1)}$ and $\overline{\LL(\AA_2)}$, and then take their union to obtain a unary HRA for $\overline{\LL_2}$.

\subsubsection*{Emptiness}
In the case of just one history, the results on
TR-VASS reachability~\cite{Schnoebelen:2010,Figueira_etal:2011} from \autoref{sec:empty} provide rather rough bounds.
It is therefore useful to do a direct analysis.
We reduce nonemptiness for unary HRAs to control-state reachability for one dimensional R-VASSs.
Our analysis below shows that the minimal path has length at most quadratic,
  from which it follows that nonemptiness has polynomial complexity.

The following result applies to any R-VASS representation for which $|Q| \log|Q| \in O(N)$,
  where $|Q|$~is the number of states of the R-VASS,
  and $N$~is the number of bits used to represent the R-VASS\null.
Note that the condition is true if all the states are listed in the R-VASS representation,
  something all reasonable representations would do.

\begin{lemma}\label{lem:1RVASS}
Control-state reachability for one dimensional R-VASSs is in~$\textsc{NL}$,
  provided that non-reset transitions increase and decrease the counter by at most~$1$.
\end{lemma}

\begin{proof}
Let $\AA=\langle Q,\delta\rangle$ be an R-VASS of dimension~$1$.
The proof relies on two observations:
\begin{description}
\item[Fact 1]
%  $\AA$ is up-monotonic:
  If $(q,i) \trdeltaa{} (q',i')$ is a configuration path of $\AA$
  then, for each $k>0$, there is a path $(q,i{+}k) \trdeltaa{} (q',i'')$ of the same length.
\item[Fact 2]
%  $\AA$ is down-monotonic:
  If $(q,i) \trdeltaa{} (q',i')$ is a configuration path of~$\AA$
    in which there are no reset transitions and the counter never becomes less than some $k>0$,
  then there is a path $(q,i{-}k) \trdeltaa{} (q',i'')$ of the same length.
\end{description}

Consider an instance $(\AA,q_0,i_0,q_F)$ of the control-state reachability problem:
  Is the state~$q_F$ reachable in~$\AA$ starting from configuration~$(q_0,i_0)$?
Let $p$ be a configuration path of minimal length
  from $(q_0,i_0)$ to some configuration whose state is~$q_F$.
%Let $p$ be a configuration path from $(q_0,i_0)$ to $(q_F,i_F)$, for some $i_F$,
%  such that $p$~is of least length among all paths leading to some $(q_F,i)$.
Let us see if a state can appear repeatedly in~$p$.
By Fact~1, $p$~is non-decreasing:
  any path segment $(q,i) \trdeltaa{} (q,i')$ can be circumvented if $i \ge i'$.
Now suppose that $p$~contains a segment $(q,i) \trdeltaa{} (q,i+k)$ for some $k>0$.
Consider the segment $(q,i+k) \trdeltaa{} (q'',i'')$
  that follows, uses only non-reset transitions, and is maximal.
By Fact~2, if the counter never becomes $<k$ in the latter segment,
  then there exists a path $(q,i) \trdeltaa{} (q'',i'')$ of the same length.
Since $p$~is minimal, this is a contradiction, and therefore the counter must become~$<k$,
  somewhere after $(q,i+k)$.
%Now suppose $(q,i{+}k)$ appears after $(q,i)$ in $p$, some $k>0$. By Fact~2 and minimality of $p$, there must be a configuration $(q',k{-}1)$ after $(q,i{+}k)$ in $p$ such that there are no reset edges between $(q,i{+}k)$ and $(q',k{-}1)$. 
Let $p'$ be the segment $(q,i+k) \trdeltaa{} (q',k-1)$.
% RG: I'm not sure the previous (implicit) claim that p' is decreasing was OK.
% (However, they it isn't needed.)
%\footnote{Here by subpath we mean a sequence of nodes from $p$ in the same order as in $p$.}
Since non-reset transitions decrease the counter by $\le 1$,
  it must be that all the values $i+k, i+k-1, \ldots, k-1$ occur in~$p'$.
When one of these values is reached for the first time,
  it must be paired with a state that was not used for the bigger values.
It follows that $(i+k)-(k-1)+1 \le |Q|$, and so $i \le |Q|-2$.
%Since $p'$ is non-decreasing, all its states are different and hence, as $p'$ has length $i{+}2$, we have $i{+}2\leq \card{Q}$, i.e.~$i\leq \card{Q}{-}2$.
This gives us a bound on the counter value of any state that can be repeated in $p$. Thus, each state can appear in $p$ at most $\card{Q}$ times.
This implies that the length of $p$ is at most $\card{Q}^2$ and that in $p$ the counter does not exceed the value $i_0+\card{Q}^2$.

We can therefore answer the instance $(\AA,q_0,i_0,q_F)$
  of the control-state reachability problem as follows.
Note first that,
  by Facts 1~and~2 and because the length of minimal reaching path is~$\le|Q|^2$,
  we can replace $i_0$ by~$\min(i_0, |Q|^2)$.
Because we only consider initial counter values $\le|Q|^2$
  and because the minimal path has length $\le|Q^2|$,
  we can store one configuration on the minimal path using $O(\log|Q|)$ bits.
Since $|Q| \log |Q| \in O(N)$, we have $\log|Q| \in O(\log N)$,
  and therefore $O(\log N)$~bits suffice to represent a configuration of the minimal path.
Finally, we note that a nondeterministic algorithm
  can guess the next configuration on the minimal path.
%Thus, we need only check $(\AA,q_0,N_0,q_F)\in\reach$ with $N_0=\min(i_0,\card{Q}^2{-}1)$.
%We do this by non-deterministically computing $\card{Q}^2$ consecutive configurations and checking whether any of them is final.
%We only store the current configuration and a counter bounded by $\card{Q}^2$.
%Thus, we require space $\log\card{Q}+\log(N_0{+}\card{Q}^2{-}1)$ for the configuration and $\log\card{Q}^2$ for the counter so, in total, less than $\log\card{Q}+\log(2\card{Q}^2)+\log\card{Q}^2$.
%Since $N=\|(\AA,q_0,N_0,q_F)\|\geq\card{Q}\cdot\log\card{Q}$, we require space $O(\log N)$.
% By Savitch's theorem, we get \textsc{DSpace}$(\log^2N)$.
%
\end{proof}

We remark that an \textsc{NL} upper bound follows
  from an analysis of the backward coverability algorithm as well.
However, the proof from above has the advantage that it is self-contained.

\medskip

We now give an upper bound for the emptiness problem of unary HRAs.
The result holds for all representations that obey several weak requirements.
Let $\langle Q,q_0,\delta,H_0,F\rangle$ be a unary HRA with $n$~registers, represented with~$N$ bits.
We require that
\begin{itemize}
\item $n \in O(N)$, which is justified because there are $>2^n$ possible labels on transitions;
\item $|\delta|\in O(N)$,
  which is justified because we expect each transition to require at least a bit;
\item $|Q| \log|Q| \in O(N)$,
  which is justified because we expect each state
  to be mentioned at least once in the representation.
  (This last point implies that $\log|Q| \in O(\log N)$.)
\end{itemize}

\begin{proposition}\label{prop:unary-empty-ub}
The emptiness problem for unary HRAs is in $\textsc{PSpace}$.
More precisely, it is in $\textsc{NSpace}(N \log N)$,
  where $N$~is the number of bits used to represent the HRA\null.
\end{proposition}

\begin{proof}
Let $\AA=\langle Q,q_0,\delta,H_0,F\rangle$ be the given unary HRA\null.
Using \autoref{prop:remove-regs-withstates},
  we build a $(1,0)$-HRA~$\AA'$ that preserves emptiness,
  has $O(B_n2^nn\log|\delta|)$ transitions, and has $O(B_n2^nn\log|Q|)$ states.
Using the construction from \autoref{lem:hra-to-trvass},
  we reduce the emptiness of~$\AA'$ to control-state reachability in an R-VASS~$\AA''$.
Specialized to our case, the construction says that
\begin{itemize}
\item
  for each transition $q \trdelta{\empti,\{1\}} q'$ in~$\AA'$,
  we include a transition $q \trdelta{+1} q'$ in~$\AA''$;
\item
  for each transition $q \trdelta{\{1\},\empti} q'$ in~$\AA'$,
  we include a transition $q \trdelta{-1} q'$ in~$\AA''$; and
\item
  for each transition $q \trdelta{\{1\}} q'$ in~$\AA'$,
  we include a transition $q \trdelta{\rm reset} q'$ in~$\AA''$.
\end{itemize}
According to \autoref{lem:1RVASS},
  the control-state reachability problem for $\AA''$ is in $\textsc{NSpace}(\log N'')$,
  where $N''$~is the number of bits used to represent~$\AA''$.
Thus, it remains to compute $N''$ as a function of~$N$.
For this, we pick one particular representation of $\AA''$, namely a list of transitions.
For such a representation we have
$
  N'' = O\bigl( B_n2^nn|\delta|\cdot \log (B_n2^nn|Q|) \bigr)
$. Thus,
\begin{align*}
  \log N'' = O(n \log n + \log |\delta| + \log (n \log n+ \log |Q|)) = O(N \log N)
\end{align*}
The last step assumes that  $n,\ |Q|\log|Q|,\ |\delta| \in O(N)$.
We require the representation of~$\AA$ to satisfy these assumptions.
\end{proof}

\begin{proposition}\label{prop:unary-empty-lb}
The emptiness problem for unary HRAs is \textsc{PSpace}-hard.
\end{proposition}
\begin{proof}
By \cite[Theorem 5.1a]{Lazic},
  the nonemptiness problem of register automata is $\textsc{PSpace}$-hard.
Register automata are a special case of unary HRAs.
\end{proof}

\begin{proposition}\label{prop:unary-empty}
The emptiness problem for unary HRAs is \textsc{PSpace}-complete.
\end{proposition}
\begin{proof}
Immediate from \autoref{prop:unary-empty-ub} and \autoref{prop:unary-empty-lb}.
\end{proof}

% >>>

%%% Local Variables:
%%% mode: latex
%%% TeX-master: "jour"
%%% End:
% vim:spell:spelllang=en_gb:fmr=<<<,>>>:

\section{Summary of Main Results}\label{sec:summary} % <<<

The theorems in this section summarize the main results proved in the previous sections.

\begin{theorem}\label{th:closure}
Languages recognised by HRAs are closed under
  union, intersection, concatenation, and Kleene star,
  but not under complementation.
Also,
\begin{itemize}
\item if resets are banned, then closure under Kleene star is lost;
\item if the number of histories is bounded, then closure under intersection is lost.
\end{itemize}
\end{theorem}

\begin{proof}
Immediate from \autoref{prop:closure}, \autoref{lem:noncomplement},
  and the closure results of \autoref{sec:weak}.
\end{proof}

\begin{theorem}\label{th:emptiness}
Deciding emptiness of an $(m,n)$-HRA has the following complexity:
\begin{enumerate}[label=\({\alph*}]
\item \textsc{NL}-complete if $m=n=0$;
\item \textsc{NP}-complete if $m=0$ and all sets labelling transitions are singletons;
\item \textsc{PSpace}-complete if $m\le1$;
\item \textsc{ExpSpace}-complete if there are no reset transitions; and
\item \textsc{Ackermann}-complete in the general case.
\end{enumerate}
\end{theorem}

\proof
(\hbox to .6 em{\rm\hss a\hss})~When $m=n=0$,
  nonemptiness is equivalent to reachability in a directed graph,
  which is a standard $\textsc{NL}$-complete problem.
\(b ~In this case, HRAs are equivalent to RAs that disallow repetitions of values in registers,
  as they were originally defined~\cite{RA1}.
  For such RAs, nonemptiness is known to be
    $\textsc{NP}$-complete~\cite[Theorem 4]{RA-NP-hard}.
\(c ~\autoref{prop:unary-empty}.
\(d ~\autoref{prop:nrHRA}.
\(e ~\autoref{prop:emptiness-general}.\qed

For universality and language inclusion, see \autoref{prop:lang-inclusion}.

% >>>

% >>>
% vim:spell:spelllang=en_gb:fmr=<<<,>>>:

\section{Connections with existing formalisms}\label{sec:connect}

We have already seen that HRAs strictly extend FRAs.
In this section, we compare HRAs with CMAs (class memory automata).
Like HRAs and FRAs,
  CMAs work on infinite alphabets,
  and have a decidable nonemptiness problem.
Implicitly, we also compare with other formalisms:
  CMAs have been shown to express the same languages as data automata%
    ~\cite[Proposition~3.7]{CMA};
  and data automata have been shown to express the same languages as
    the two-variable fragment of
    existential monadic second order logic
      with data equality, position successor, and class successor%
        ~\cite[Proposition~14]{DA}.

%In this section we shall draw connections between HRAs and
%an automata model over infinite alphabets at the limits of decidability, called \emph{Data Automata (DA)}, introduced in~\cite{DA} in the context of XML theory. 
%DAs operate on \emph{data words}, i.e.~over finite sequences of elements from $\SS\times\names$, where $\SS$ is a finite set of \emph{data tags} and $\names$ is an infinite set of \emph{data values} (but we shall call them \emph{names}). 
%A DA operates in two stages which involve a transducer automaton and a finite-state automaton respectively. Both automata operate on the tag projection of the input, with the second automaton focussing on tags paired with the same name.
%
%For the rest of our discussion we shall abuse data words and treat them simply as strings of names, neglecting data tags. This is innocuous since there are straightforward translations between the two settings.\footnote{A string of names is the same as a data word over a singleton set of data tags; while data tags can be simulated by names in registers of the initial configuration which do not get moved nor copied during the computation.}
%An equivalent formulation of DAs which is closer to our framework is the following~\cite{CMA}.

\begin{definition}
A \boldemph{Class Memory Automaton (CMA)} is a tuple $\AA=\langle Q,q_0,\phi_0,\delta,F_1,F_2\rangle$ where $Q$ is a finite set of states, $q_0\in Q$ is initial, $F_1\subseteq F_2\subseteq Q$ are sets of final states and the transition relation is of type \
$
\delta\subseteq Q\times (Q\cup\{\bot\})\times Q
$. Moreover, $\phi_0$ is an initial \emph{class memory function}, that is, a function $\phi:\names\to Q\cup\{\bot\}$ with finite domain ($\{\,a \mid \phi(a)\not=\bot\,\}$ is finite).
\end{definition}

The semantics of a CMA $\AA$ is given as follows. Configurations of $\AA$ are pairs of the form $(q,\phi)$, where $q\in Q$ and $\phi$ a class memory function. The configuration graph of $\AA$ is constructed by setting $(q,\phi)\trdelta{a}(q',\phi')$ just if there is $(q,\phi(a),q')\in\delta$ and $\phi'=\phi[a\mapsto q']$. The initial configuration is $(q_0,\phi_0)$, while a configuration $(q,\phi)$ is accepting just if $q\in F_1$ and, for all $a\in\names$, $\phi(a)\in F_2\cup\{\bot\}$.

Thus, CMAs resemble HRAs in that they store input names in ``histories", only that histories are identified with states: for each state $q$ there is a corresponding history $q$ (note notation overloading), and a transition which accepts a name $a$ and leads to a state $q$ must store $a$ in the history $q$. Moreover, each name appears in at most one history (hence the type of $\phi$), while the finality conditions for configurations allow us to impose that, at the end, all names must appear in specific histories, if they appear in any.
For instance,
  the language $\overline{\LL_4}$ of \autoref{ex:complement},
  which we know cannot be recognized by HRAs (\autoref{lem:noncomplement}),
  can be recognized by the following CMA on the left (with $F_1=F_2=\{q_0\}$).
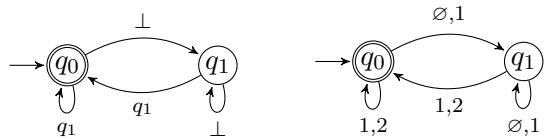
\begin{center}
\begin{tikzpicture}[automaton]
\node[state,initial,accepting] (q0) {$q_0$};
\node[state] (q1) [right=of q0] {$q_1$};
\path[transition]
  (q0) edge[bend left] node[above]{$\scriptstyle \bot$} (q1)
  (q1) edge[bend left] node[below]{$\scriptstyle q_1$} (q0)
  (q0) edge[loop below] node{$\scriptstyle q_1$} ()
  (q1) edge[loop below] node{$\scriptstyle \bot$} ();
\end{tikzpicture}\qquad
\begin{tikzpicture}[automaton]
\node[state,initial,accepting] (q0) {$q_0$};
\node[state] (q1) [right=of q0] {$q_1$};
\path[transition]
  (q0) edge[bend left] node[above]{$\scriptstyle \empti,1$} (q1)
  (q1) edge[bend left] node[below]{$\scriptstyle 1,2$} (q0)
  (q0) edge[loop below] node{$\scriptstyle 1,2$} ()
  (q1) edge[loop below] node{$\scriptstyle \empti,1$} ();
\end{tikzpicture}%\vspace{-2.75mm}
\end{center}
Each name is put in history $q_1$ when seen for the first time, and in history~$q_0$ when seen for the second time. The automaton accepts if all its names are in $q_0$. This latter condition is what makes the essential difference to HRAs, namely the capability to check where the names reside for acceptance. For example, the HRA on the right above would accept the same language it we were able to impose the condition that accepting configurations $(q,H)$ satisfy $a\in H\at\{2\}$ for all names $a\in \bigcup_i H(i)$.
Note though that, extending HRAs with such finality conditions would render their nonemptiness problem reducible from reachability of R-VASS (i.e.\ the question whether a specific state \emph{and} counter content can be reached), a problem known to be undecidable~\cite{ArakiK76}.

The above example proves that HRAs cannot express the same languages as CMAs. Conversely, as shown in~\cite[Proposition~7.2]{CMA}, the fact that CMAs lack resets does not allow them to express languages like, for example, $\LL_1$. 
{As a result, the languages expressed by CMAs are closed under intersection, union and concatenation, but not under Kleene star.}
%\rgtext{(In contrast, recall that HRAs are closed under Kleene star.)}
In the latter sections of~\cite{CMA} several extensions of CMAs are considered, one of which does involve resets. However, the resets considered there do not seem directly comparable to the reset capability of HRAs.

On the other hand, a direct comparison can be made with non-reset HRAs. We already saw in \autoref{prop:nrHRA-regs} that, in the latter idiom, histories can be used for simulating register behaviour. In the absence of registers, CMAs differ from non-reset HRAs solely in their constraint of relating histories to states (and their termination behaviour, which is more expressive). As the latter can be easily counterbalanced by obfuscating the set of states, we obtain the following.

\begin{proposition}
For each non-reset HRA $\AA$ there is a CMA \hbox{$\AA'$ such that $\LL(\AA)=\LL(\AA')$.}
\end{proposition}

\cutout{
  It might be good to also comment on complexity, but the CMA paper is light on this.
  Should we also say something about nested CMAs?
}

\cutout{It is straightforward to show that CMAs extend in expressiveness non-reset HRAs which contain no registers.
Since CMAs can also express registers, it should be possible to apply the colouring technique and show that 
for any non-reset HRA $\AA$ there is a CMA $\AA'$ such that $\LL(\AA)=\LL(\AA')$.
The above could be achieved either directly, or by showing that non-reset HRAs can be equivalently decomposed into machines containing exclusively histories or registers, and use the fact that CMAs extend both in expressiveness.}

%%% Local Variables: 
%%% mode: latex
%%% TeX-master: "jour"
%%% End: 
% vim:spell:spelllang=en_gb:

\section{Further directions}
Our goal is to apply automata with histories in static and runtime verification. 
For static verification, the complexity results derived in this paper may seem discouraging at first. However, they are based on very specific representations of hard problems; in practice, we expect programs to yield automata of simpler complexities. Experience with tools based on coverability of TR-VASSs, like e.g.~BFC~\cite{BFC}, positively testify in that respect.
Another solution, already pursued herein, is to explore constrained versions of our machines.
A specific such variant we envisage to consider is one with restricted resets, in analogy to e.g.~\cite{zero}.
In a related direction, we aim to look at abstractions that would allow us to attack the model-checking problem for these automata, and also look at temporal logics that capture part or all of the expressivity of HRAs.

In this work we examined nondeterministic automata but did not look at alternating variants. This is justified by the undecidability of universality already at the level of register automata.
However, if one is willing to restrict the number of registers and histories, there may still be room for decidability.
In the case of register automata, it has been shown~\cite{Lazic}
  that alternating register automata with one register are decidable for emptiness,
  and become undecidable at two registers.
While these automata cannot capture languages that inherently require more than one register, they can use alternation to express name freshness and e.g.\ capture the languages $\LL_0,\LL_2$ of the Introduction, and also a variant of $\LL_1$ which uses constants for tokenizing the input (instead of $a_0$). It would be useful to examine whether a similar restriction can yield decidable alternating HRAs, and what would their expressivity be.
Finally, a problem left open here is decidability and complexity of bisimilarity.
(In a private communication,
  Piotrek Hofman sketched a proof that bisimilarity is decidable.)
%\rg{We should maybe mention that we were told in a private communication that
%  bisimilarity is decidable.
%I still didn't check that proof $\ldots$}
%Although it is known that bisimilarity is undecidable for Petri nets~\cite{Jancar}, the version which seems of relevance towards an undecidability argument for HRAs is that of \emph{visibly} counter automata with labels, i.e.~automata which accept labels at each transition, and the action of each transition is determined by its label. The latter problem is not known to be decidable.

\frenchspacing
\bibliography{biblio}{}

\begin{thebibliography}{10}

\bibitem{nom3}
S.~Abramsky, D.~R. Ghica, A.~S. Murawski, C.-H.~L. Ong, and I.~D.~B. Stark.
\newblock Nominal games and full abstraction for the nu-calculus.
\newblock In {\em Logic in Computer Science (LICS)}, 2004.

\bibitem{ArakiK76}
T.~Araki and T.~Kasami.
\newblock Some decision problems related to the reachability problem for
  {Petri} nets.
\newblock {\em Theoretical Computer Science (TCS)}, 1976.

\bibitem{Java2}
M.~Faouzi Atig, A.~Bouajjani, and S.~Qadeer.
\newblock Context-bounded analysis for concurrent programs with dynamic
  creation of threads.
\newblock {\em Logical Methods in Computer Science (LMCS)}, 2011.

\bibitem{nom2}
N.~Benton and B.~Leperchey.
\newblock Relational reasoning in a nominal semantics for storage.
\newblock In {\em Typed Lambda Calculi and Applications (TLCA)}, 2005.

\bibitem{CMA}
H.~Bj{\"o}rklund and T.~Schwentick.
\newblock On notions of regularity for data languages.
\newblock {\em Theoretical Computer Science (TCS)}, 2010.

\bibitem{WNom1}
M.~Bojanczyk, L.~Braud, B.~Klin, and S.~Lasota.
\newblock Towards nominal computation.
\newblock In {\em Principles of Programming Languages (POPL)}, 2012.

\bibitem{DA}
M.~Bojanczyk, C.~David, A.~Muscholl, T.~Schwentick, and L.~Segoufin.
\newblock Two-variable logic on data words.
\newblock {\em Transactions on Computational Logic (TOCL)}, 2011.

\bibitem{WNom3}
M.~Bojanczyk, B.~Klin, and S.~Lasota.
\newblock Automata theory in nominal sets.
\newblock {\em Logical Methods in Computer Science (LMCS)}, 2014.

\bibitem{WNom2}
M.~Bojanczyk, B.~Klin, S.~Lasota, and S.~Torunczyk.
\newblock Turing machines with atoms.
\newblock In {\em Logic in Computer Science (LICS)}, 2013.

\bibitem{Java1}
A.~Bouajjani, S.~Fratani, and S.~Qadeer.
\newblock Context-bounded analysis of multithreaded programs with dynamic
  linked structures.
\newblock In {\em Computer Aided Verification (CAV)}, 2007.

\bibitem{backward-cover-analysis}
L.~Bozzelli and P.~Ganty.
\newblock Complexity analysis of the backward coverability algorithm for
  {VASS}.
\newblock In {\em Reachability Problems (RP)}, 2011.

\bibitem{Rackoff}
Rackoff C.
\newblock The covering and boundedness problems for vector addition systems.
\newblock {\em Theoretical Computer Science (TCS)}, 1978.

\bibitem{NCMA}
C.~Cotton{-}Barratt, A.~S. Murawski, and C.{-}H.~Luke Ong.
\newblock Weak and nested class memory automata.
\newblock In {\em Language and Automata Theory and Applications (LATA)}, 2015.

\bibitem{NDA}
N.~Decker, P.~Habermehl, M.~Leucker, and D.~Thoma.
\newblock Ordered navigation on multi-attributed data words.
\newblock In {\em Concurrency Theory (CONCUR)}, 2014.

\bibitem{Lazic}
S.~Demri and R.~Lazi\'c.
\newblock {LTL} with the freeze quantifier and register automata.
\newblock {\em Transactions on Computational Logic (TOCL)}, 2009.

\bibitem{Figueira_etal:2011}
D.~Figueira, S.~Figueira, S.~Schmitz, and P.~Schnoebelen.
\newblock Ackermannian and primitive-recursive bounds with {Dickson's} lemma.
\newblock In {\em Logic in Computer Science (LICS)}, 2011.

\bibitem{zero}
A.~Finkel and A.~Sangnier.
\newblock Mixing coverability and reachability to analyze {VASS} with one
  zero-test.
\newblock In {\em Current Trends in Theory and Practice of Computer Science
  (SOFSEM)}, 2010.

\bibitem{nom}
M.~J. Gabbay and A.~M. Pitts.
\newblock A new approach to abstract syntax with variable binding.
\newblock {\em Formal Aspects of Computing}, 2002.

\bibitem{oded-complexity}
O.~Goldreich.
\newblock {\em Computational Complexity: A Conceptual Perspective}.
\newblock Cambridge University Press, 2008.

\bibitem{TOPL}
R.~Grigore, D.~Distefano, R.~L. Petersen, and N.~Tzevelekos.
\newblock Runtime verification based on register automata.
\newblock In {\em Tools and Algorithms for the Construction and Analysis of
  Systems (TACAS)}, 2013.

\bibitem{nom1}
A.~Jeffrey and J.~Rathke.
\newblock Towards a theory of bisimulation for local names.
\newblock In {\em Logic in Computer Science (LICS)}, 1999.

\bibitem{BFC}
A.~Kaiser, D.~Kroening, and T.~Wahl.
\newblock Efficient coverability analysis by proof minimization.
\newblock In {\em Concurrency Theory (CONCUR)}, 2012.

\bibitem{RA1}
M.~Kaminski and N.~Francez.
\newblock Finite-memory automata.
\newblock {\em Theoretical Computer Science (TCS)}, 1994.

\bibitem{nom4}
J.~Laird.
\newblock A fully abstract trace semantics for general references.
\newblock In {\em Automata, Languages, and Programming (ICALP)}, 2007.

\bibitem{Lipton}
R.~J. Lipton.
\newblock The reachability problem requires exponential space.
\newblock Technical report, Yale University, 1976.

\bibitem{Manuel_Ramanujam:2011}
A.~Manuel and R.~Ramanujam.
\newblock Class counting automata on datawords.
\newblock {\em Foundations of Computer Science (IJFCS)}, 2011.

\bibitem{HDA}
U.~Montanari and M.~Pistore.
\newblock An introduction to history dependent automata.
\newblock {\em Electronic Notes in Theoretical Computer Science (ENTCS)}, 1997.

\bibitem{IMJ2}
A.~S. Murawski, S.~J. Ramsay, and N.~Tzevelekos.
\newblock Game semantic analysis of equivalence in {IMJ}.
\newblock In {\em Automated Technology for Verification and Analysis (ATVA)},
  2015.

\bibitem{ML1}
A.~S. Murawski and N.~Tzevelekos.
\newblock Algorithmic nominal game semantics.
\newblock In {\em European Symposium on Programming (ESOP)}, 2011.

\bibitem{ML2}
A.~S. Murawski and N.~Tzevelekos.
\newblock Algorithmic games for full ground references.
\newblock In {\em Automata, Languages, and Programming (ICALP)}, 2012.

\bibitem{RA2}
F.~Neven, T.~Schwentick, and V.~Vianu.
\newblock Finite state machines for strings over infinite alphabets.
\newblock {\em Transactions on Computational Logic (TOCL)}, 2004.

\bibitem{nu}
A.~M. Pitts and I.~Stark.
\newblock On the observable properties of higher order functions that
  dynamically create local names, or: What's new?
\newblock In {\em Mathematical Foundations of Computer Science (MFCS)}, 1993.

\bibitem{RA-NP-hard}
H.~Sakamoto and D.~Ikeda.
\newblock Intractability of decision problems for finite-memory automata.
\newblock {\em Theoretical Computer Science (TCS)}, 2000.

\bibitem{algo-wqo}
S.~Schmitz and P.~Schnoebelen.
\newblock Algorithmic aspects of {WQO} theory.
\newblock Lecture Notes $\langle$cel-00727025v2$\rangle$, 2012.

\bibitem{Schnoebelen:2010}
P.~Schnoebelen.
\newblock Revisiting {Ackermann}-hardness for {Lossy Counter Machines} and
  {Reset Petri Nets}.
\newblock In {\em Mathematical Foundations of Computer Science (MFCS)}, 2010.

\bibitem{Segoufin_overview}
L.~Segoufin.
\newblock Automata and logics for words and trees over an infinite alphabet.
\newblock In {\em Computer Science Logic (CSL)}, 2006.

\bibitem{Stark:PhD}
I.~D.~B. Stark.
\newblock {\em Names and Higher-Order Functions}.
\newblock PhD thesis, University of Cambridge Computing Laboratory, 1995.

\bibitem{FRA}
N.~Tzevelekos.
\newblock Fresh-register automata.
\newblock In {\em Principles of Programming Languages (POPL)}, 2011.

\bibitem{CONF}
N.~Tzevelekos and R.~Grigore.
\newblock History-register automata.
\newblock In {\em Foundations of Software Science and Computation Structures
  (FoSSaCS)}, 2013.

\end{thebibliography}
\bibliographystyle{plain}

\ifdraft{
\newpage
{\Huge Stuff being rewritten:}

\input{old-emptiness2.tex}
\input{old-regs.tex}
\input{old-weakening.tex}
}{}

\end{document}

%BBBBBBBBBBBBBBBBBBBBBBBBBBBBBBBBBBBBBBBBBBBBBBBBBBBBBBBBBBBBBBBBBBBBBBBBBBBBBBBBBBBBBBBBBBBBBBBBBBBBBBBBBBBBBBBBBBBBBBBBBBBBBBBBBBBBBB